\documentclass{article}

\newtheorem{theorem}{Theorem}[section]
\newtheorem{corollary}{Corollary}[section]

\newtheorem{proposition}{Proposition}[section]
\newtheorem{definition}{Definition}[section]

\newtheorem{example}{Example}[section]
\newtheorem{proof}{Proof}[section]

\usepackage{arxiv}
\usepackage{subcaption}
\usepackage{graphicx}
\usepackage[utf8]{inputenc} 
\usepackage[T1]{fontenc}    
\usepackage{hyperref}       
\usepackage{url}            
\usepackage{booktabs}       
\usepackage{amsfonts}       
\usepackage{amsmath}
\usepackage{amssymb}
\usepackage{nicefrac}       
\usepackage{microtype}      
\usepackage{cleveref}       
\usepackage{graphicx}
\usepackage{doi}

\title{An Algebraic Construction Technique for Codes over Hurwitz Integers}

\author{{\includegraphics[scale=0.06]{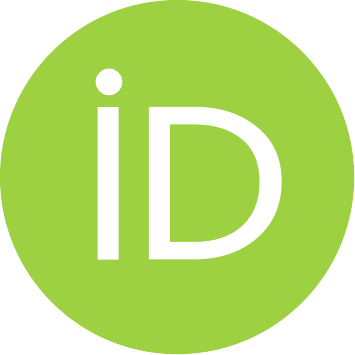}\hspace{1mm}Ramazan ~DURAN} \\
	Department of Mathematics\\
	Afyon Kocatepe University\\
	Afyonkarahisar, TR 03200 \\
	\And
	{\includegraphics[scale=0.06]{orcid.pdf}\hspace{1mm}Murat~GUZELTEPE} \\
	Department of Mathematics\\
	Sakarya University\\
	Sakarya, TR 54187 \\
}

\renewcommand{\shorttitle}{Algebraic Construction Technique}

\begin{document}
\maketitle

\begin{abstract}
Let $\alpha$ be a prime Hurwitz integer. $\mathcal{H}_{\alpha}$, which is the set of residual class with respect to related modulo function in the rings of Hurwitz integers, is a subset of $\mathcal{H},$ which is the set of all Hurwitz integers. We consider left congruent modulo $\alpha$ and, the domain of related modulo function in this study is $\mathbb{Z}_{N(\alpha)},$ which is residual class ring of ordinary integers with $N(\alpha)$ elements, which is the norm of prime Hurwitz integer $\alpha.$ In this study, we present an algebraic construction technique, which is a modulo function formed depending on two modulo operations, for codes over Hurwitz integers. Thereby, we obtain the residue class rings of Hurwitz integers with $N(\alpha)$ size. In addition, we present some results for mathematical notations used in two modulo functions, and for the algebraic construction technique formed depending upon two modulo functions. Moreover, we presented graphs obtained by graph layout methods, such as spring, high-dimensional, and spiral embedding,  for the set of the residual class obtained with respect to the related modulo function in the rings of Hurwitz integers.
\end{abstract}


\section{Introduction} \label{sec1}

The first study on codes over high-dimensional algebraic structures such as Gauss integers, Einstein-Jacobi integers, quaternions, Lipschitz integers, and Hurwitz integers was got by Huber. In \cite{bib1}, Huber showed how block codes over Gaussian integers could be used for coding over two-dimensional algebraic structure. The set of Gaussian integers that denoted by $\mathbb{Z}[i]$ is shown by $\mathbb{Z}[i]=\{ \alpha = \alpha_{1} + \alpha_{2}i : \alpha_{1} , \alpha_{2} \in \mathbb{Z} \ \text{where} \ i^2=-1 \}.$ In this study, we use $\mathcal{G}$ notation instead of $\mathbb{Z}[i].$ $\alpha_{1}$ is the real part, and $\alpha_{2}i$ is the imaginary part. The conjugate of Gaussian integer $\alpha$ is $\overline{\alpha}=\alpha_{1}-\alpha_{2}i.$  The norm of Gaussian integer $\alpha$ is $N(\alpha)=\alpha \overline{\alpha}=\alpha^{2}_{1}+\alpha^{2}_{2}.$ A Gaussian integer is called a prime Gaussian integer if its norm is a prime integer. In \cite{bib1}, Huber considered Gaussian integers such that $p\equiv1 \mod N(\alpha)$ where $p$ is a prime integer and, presented a technique known as the modulo function to construct block codes over Gaussian integers, a two-dimensional vector space. In this way, Huber constructed one Mannheim error-correcting (OMEC) codes over Gaussian integers fields. The modulo function $\mu:\mathcal{G}\rightarrow\mathcal{G}_{\alpha}$ is defined according to
\begin{equation}\label{equu1}
\mu(\xi)=\xi \mod \alpha = \eta = \xi - \left[ \frac{\xi \overline\alpha}{\alpha \overline\alpha} \right] \alpha
\end{equation}
where $\xi \in \mathcal{G}$ and $\alpha$ is a prime Gaussian integer \cite{bib1}. Here $\mathcal{G}_{\alpha}$ is the set of residual class with respect to modulo $\alpha.$ We give the definition of mathematical notation $[\cdot]$ in the following section. In a similar technique, in \cite{bib2}, Huber showed how block codes over Eisenstein-Jacobi integers could be used for coding over two-dimensional vector space. In \cite{bib3}, Freudenberger et al. presented new coding techniques for codes over Gaussian integers by using the modulo function in \cite{bib1}.

Quaternions, a four-dimensional algebraic structure, are a number system that expands over complex numbers. They form $\alpha=\alpha_{1}+\alpha_{2}i+\alpha_{3}j+\alpha_{4}k$ where $\alpha_{1}, \alpha_{2}, \alpha_{3}, \alpha_{4} \in \mathbb{R}$ such that $i^{2}=j^{2}=k^{2}=-1,$ and $ij=-ji=k,$ $jk=-kj=i,$ $ki=-ik=j.$ $\alpha_{1}, \alpha_{2}, \alpha_{3},$ and $\alpha_{4}$ are the components of quaternion $\alpha$. In \cite{bib4}, Ozen and Guzeltepe studied codes over some finite fields by using quaternion integers, which have the commutative property. In \cite{bib5}, Ozen and Guzeltepe studied cyclic codes over some finite quaternion integer rings. They considered the quaternion integers which have the commutative property. In \cite{bib6}, Shah and Rasool established that; over quaternion integers, for a given $n$ length cyclic code there exists a cyclic code of length $2n-1.$ A lipschitz integer is a quaternion whose components are all integers. In \cite{bib7}, Freudenberger et al. presented new block codes over Lipschitz integers. They consider primitive Lipschitz integers, which is the greatest common divisor of components is one. A quaternion is called a Hurwitz integer if $\alpha_{1}, \alpha_{2}, \alpha_{3}, \alpha_{4}$ are either in $\mathbb{Z}$ or in $\mathbb{Z}+\frac{1}{2}.$ Also, Hurwitz integers forms a subring of the ring of all quaternions. If the norm of a Hurwitz integer is a prime integer, then it is called a prime Hurwitz integer. In this study, we consider prime Hurwitz integers. In \cite{bib8}, Ozen and Guzeltepe obtained new classes of linear codes over Hurwitz integers. They considered Hurwitz integers, which have the commutative property. In \cite{bib9}, Rohweder et al. presented a new algebraic construction technique for codes over Hurwitz integers that is inherently accompanied by a respective modulo operation. They were considered that the domain of related modulo function in \cite{bib9} is $\mathbb{Z}_{N(\alpha)} \times \mathbb{Z}_{N(\alpha)}.$ Therefore, they constructed new sets of residual class, which has $N^{2}(\alpha)$ elements, with respect to the related modulo technique in the rings of Hurwitz integers, using the primitive Lipschitz integers. Other studies on codes over Hurwitz integers were presented in papers to \cite{bib10}, \cite{bib11}, and \cite{bib12}.

The presentation of our results is organized as follows: In the\hyperref[sec2]{next section}, we give the necessary fundamental definitions used throughout this paper. In \hyperref[sec3]{section \ref{sec3}}, we present an algebraic construction technique, which is a modulo function formed depending upon two modulo functions such that the domain of it is $\mathbb{Z}_{N(\alpha)}.$  Thereby, we construct new set of residual class with $N(\alpha)$ elements, with respect to the related modulo function in the rings of Hurwitz integers. In \hyperref[sec4]{section \ref{sec4}}, we give some of the natural results for mathematical notations in two modulo functions and the modulo function. In \hyperref[sec5]{section \ref{sec5}}, we give examples for the algebraic construction technique given in this study. In \hyperref[sec6]{section \ref{sec6}}, we investigate the code rates of some codes over Hurwitz integers. In \hyperref[sec7]{section \ref{sec7}}, we give the points and graph on $2D$ (two-dimensional) and $3D$ (three-dimensional) the elements of the residual class obtained with respect to the related modulo function in prime Hurwitz integers using graph layout methods such as spring, high-dimensional, and spiral embedding. We conclude this paper with \hyperref[sec8]{section \ref{sec8}}.

\section{Preliminaries} \label{sec2}

We begin with some basic definitions.

\begin{definition} \label{def1}
The Hamilton quaternion algebra over $R$, denoted by $\mathbb{H}(\mathbb{R})$, is the associative unital algebra given by the following presentation:
\begin{itemize}
  \item[i.] $\mathbb{H}(\mathbb{R})$ is free $R-$module over the symbols $1,$ $i,$ $j,$ $k;$ that is, $\mathbb{H}(\mathbb{R}) = \{ \alpha = \alpha_{1}+\alpha_{2}i+\alpha_{3}j+\alpha_{4}k : \alpha_{1}, \alpha_{2}, \alpha_{3}, \alpha_{4} \in \mathbb{R} \};$
  \item[ii.] $1$ is the multiplicative unit;
  \item[iii.] $ i^{2}=j^{2}=k^{2}=-1;$
  \item[iv.] $ ij=-ji=k;$ $jk=-kj=i;$ $ki=-ik=j$ \cite{bib13}.
\end{itemize}
\end{definition}

Note that, $\alpha_{1}$ is the real part, $\alpha_{2}i+\alpha_{3}j+\alpha_{4}k$ is imaginary part, and $\alpha_{1},\alpha_{2},\alpha_{3},$ and  $\alpha_{4}$ are the components. If $\alpha_{1}, \alpha_{2}, \alpha_{3}, \alpha_{4} \in \mathbb{Z},$ then $\alpha$ is called a quaternion integer. The commutative property of multiplication does not hold over quaternion integers, in general. The set of quaternion integers which have the commutative property that denoted by $\mathcal{H}^{c}$ is shown by $\mathcal{H}^{c}= \{\alpha_{1} +\alpha_{2} ( i+j+k ): \alpha_{1},\alpha_{2} \in \mathbb{Z} \}.$ In the other words, if the imaginary parts of quaternion integers are parallel to each other, then their product is commutative.

\begin{definition}\label{def2}
$\alpha  = \alpha_{1} + \alpha_{2}i + \alpha_{3}j + \alpha_{4}k$ is called a Hurwitz integer just if either $ \alpha_{1}, \alpha_{2}, \alpha_{3}, \alpha_{4} \in \mathbb{Z}$ or $ \alpha_{1}, \alpha_{2}, \alpha_{3}, \alpha_{4} \in \mathbb{Z}+\frac{1}{2}. $ The set of all Hurwitz integers that denoted by $\mathcal{H}$ is shown by
\begin{equation} \label{equ1}
\begin{array}{lll}
 \mathcal{H} & = & \left\{ { \alpha_{1} + \alpha_{2}i + \alpha_{3}j + \alpha_{4}k  : \alpha_{1}, \alpha_{2}, \alpha_{3}, \alpha_{4}  \in \mathbb{Z} \ or \ \alpha_{1}, \alpha_{2}, \alpha_{3}, \alpha_{4}  \in \mathbb{Z}+\frac{1}{2} } \right\} \hfill \\
 & = & \mathcal{H}(\mathbb{Z}) \bigcup \mathcal{H}(\mathbb{Z}+\frac{1}{2}).
\end{array}
\end{equation}
\end{definition}

\begin{example}\label{ex1}
\begin{itemize}
  \item[i.] $\pm 1 \pm\frac{1}{2}i \pm \frac{1}{2}j \pm \frac{1}{2}k$ is not a Hurwitz integer.
  \item[ii.] $\pm \frac{1}{2} \pm \frac{1}{2}j$ is not a Hurwitz integer.
  \item[iii.] $\pm \frac{1}{2} \pm \frac{1}{2}i \pm \frac{1}{2}j \pm \frac{1}{2}k$ is a Hurwitz integer.
\end{itemize}
\end{example}

The ring of Hurwitz integers forms a subring of the ring of all quaternions because of closed under multiplication and addition. The conjugate of $\alpha$ is $\overline \alpha = \alpha_{1} - \alpha_{2}i - \alpha_{3}j - \alpha_{4}k,$ the norm of $\alpha$ is $ N\left( \alpha  \right) = \alpha  \cdot \overline \alpha  = \alpha _{1}^2 + \alpha _{2}^2 + \alpha _{3}^2 + \alpha _{4}^2,$ and the inverse of $\alpha$ is ${\alpha ^{ - 1}} = \frac{{\overline \alpha }}{{N\left( \alpha  \right)}}$ where $N(\alpha) \neq 0.$

\begin{definition}\label{def4} Let $\alpha  = \alpha_{1} + \alpha_{2}i + \alpha_{3}j + \alpha_{4}k$ be a Hurwitz integers such that $\alpha_{1}, \alpha_{2}, \alpha_{3}, \alpha_{4} \in \mathbb{Z}.$ $\alpha$ is called a primitive Hurwitz integer just if the greatest common divisor ($\mathrm{gcd}$) of $ \alpha_{1}, \alpha_{2}, \alpha_{3},$ and $\alpha_{4}$ is one, namely, $(\alpha_{1}, \alpha_{2}, \alpha_{3}, \alpha_{4})=1.$
\end{definition}

\begin{definition}\label{def5} Let $\alpha  = \alpha_{1} + \alpha_{2}i + \alpha_{3}j + \alpha_{4}k$ be a Hurwitz integers such that $\alpha_{1}=\beta_{1}+\frac{1}{2}, \alpha_{2}=\beta_{2}+\frac{1}{2}, \alpha_{3}=\beta_{2}+\frac{1}{2}, \alpha_{4}=\beta_{4}+\frac{1}{2}$ where $\beta_{1}, \beta_{2}, \beta_{3}, \beta_{4} \in \mathbb{Z}.$  $\alpha$ is called a primitive Hurwitz integer just if the greatest common divisor ($\mathrm{gcd}$) of $ 2\beta_{1}+1, 2\beta_{2}+1, 2\beta_{3}+1,$ and $2\beta_{4}+1$ is one, namely, $(2\beta_{1}+1, 2\beta_{2}+1, 2\beta_{3}+1, 2\beta_{3}+1)=1.$
\end{definition}

\begin{definition} \label{def6}
$\alpha  = \alpha_{1} + \alpha_{2}i + \alpha_{3}j + \alpha_{4}k$ is called a prime Hurwitz integer if its norm is a prime integer.
\end{definition}

From \hyperref[def4]{definition \ref{def4}}, \hyperref[def5]{definition \ref{def5}} and \hyperref[def6]{definition \ref{def6}}, we say that every prime Hurwitz integer is a primitive Hurwitz integer but converse may not be true. Note that we consider prime Hurwitz integers in this study.

\begin{definition} \label{def7}
Let $\alpha$ be a Hurwitz integer. If there exists $\lambda \in \mathcal{H}$ such that $q_{1} - q_{2} = \lambda \pi,$ then $q_{1}, q_{2} \in \mathcal{H}$ are said to be right congruent modulo $\alpha.$ This relation is denoted by $q_{1} \equiv_{r} q_{2}.$ This relation $q_{1} \equiv_{r} q_{2}$ is an equivalence relation. The elements in the right ideal
\begin{equation}\label{equ2}
\langle \alpha \rangle = \{ \lambda \alpha : \lambda \in \mathcal{H} \}
\end{equation}
define a normal subgroup of the additive group of the ring $\mathcal{H}.$ The set of cosets to $\langle \alpha \rangle$ in $\mathcal{H}$ defines Abelian group denoted by $\mathcal{H}_{\alpha}=\mathcal{H}\diagup \langle \alpha \rangle.$ Analogous results are valid for left congruences modulo $\alpha$ \cite{bib8}.
\end{definition}

Note that we use left congruences modulo $\alpha$ in this study. The following section is going to present the algebraic construction technique, formed depending on two modulo functions, used to construct block codes of $N(\alpha)$ size over Hurwitz integers. However, we close this section by giving the definitions of two mathematical notations used in the modulo functions before going to the next section.

\begin{definition} \label{def8}
A notation for rounding to the nearest integer is denoted by $\lfloor \cdot \rceil.$ It is defined as rounding a rational number to the integer closest to it. All components in a quaternion are separately rounding to the integer closest to it.
\end{definition}

\begin{definition} \label{def9}
A notation for rounding to the nearest half-integer is denoted by $\lfloor \lfloor \cdot \rceil \rceil.$ It is defined as rounding a rational number to the half-integer closest to it. All components in a quaternion are separately rounding to the half-integer closest to it \cite{bib11}.
\end{definition}

Note that if the component is negative or zero, then the rounding is done in the direction up. Otherwise the rounding is done in the direction down in this study.

\begin{example}\label{ex2}
\begin{itemize}
  \item[i.] $\lfloor \frac{3}{2}\rceil = 1,$ $ \lfloor -\frac{1}{2} \rceil = 0,$ $ \lfloor -\frac{3}{2} \rceil = -1,$  and so on.
  \item[ii.] $\lfloor \lfloor 0 \rceil\rceil = \frac{1}{2}$ where $0 \in \mathbb{Z},$ $ \lfloor\lfloor -1 \rceil\rceil = -\frac{1}{2},$ $ \lfloor\lfloor 1 \rceil\rceil = \frac{1}{2},$  and so on.
  \item[iii.] $\lfloor \lfloor 0 \rceil\rceil = \frac{1}{2}+\frac{1}{2}i+\frac{1}{2}j+\frac{1}{2}k$ where $0$ is a quaternion.
\end{itemize}
\end{example}

\section{Algebraic Construction Technique} \label{sec3}

This section presents an algebraic construction technique formed depending on two modulo functions used to construct codes of $N(\alpha)$ size over Hurwitz integers.

\begin{definition}\label{def10}
Let $\alpha$ be a prime Hurwitz integer. The modulo function $\mu: \mathbb{Z}_{N(\alpha)} \rightarrow \mathcal{H}_{\alpha}$ is defined by
\begin{equation}\label{equ3}
\mu_{\alpha}(z) = \mathrm{min} \{ \mu_{\alpha}^{(1)}(z), \mu_{\alpha}^{(2)}(z) \}
\end{equation}
where
\begin{equation}\label{equ4}
\mu_{\alpha}^{(1)}(z) = z \mod \alpha = z - \alpha \lfloor \frac{\overline{\alpha} z}{N(\alpha)} \rceil
\end{equation} and
\begin{equation}\label{equ5}
\mu_{\alpha}^{(2)}(z) = z \mod \alpha = z - \alpha \lfloor \lfloor \frac{\overline{\alpha} z}{N(\alpha)} \rceil \rceil
\end{equation} such that
\begin{equation}\label{equ6}
    \mathrm{min} \{ \mu_{\alpha}^{(1)}(z), \mu_{\alpha}^{(2)}(z) \} = \begin{cases}
                                                                  \mu_{\alpha}^{(1)}(z), & \mbox{if } N(\mu_{\alpha}^{(1)}(z)) < N(\mu_{\alpha}^{(2)}(z)) \\
                                                                  \mu_{\alpha}^{(1)}(z), & \mbox{if } N(\mu_{\alpha}^{(1)}(z)) = N(\mu_{\alpha}^{(2)}(z)) \mbox{ where z is an even integer} \\
                                                                  \mu_{\alpha}^{(2)}(z), & \mbox{otherwise}.
                                                                \end{cases}
\end{equation}
where $z \in \mathbb{Z}_{N(\alpha)}.$ Here $\mathbb{Z}_{N(\alpha)}$ is the well-known residual class ring of ordinary integers with $N(\alpha)$ elements, $\mathcal{H}_{\alpha}$ is the residual class ring of prime Hurwitz integer $\alpha,$ and $\mu_{\alpha}^{(1)}(z)$ and $\mu_{\alpha}^{(2)}(z)$ is given remainder of $ z (z \in \mathbb{Z}_{N(\alpha)})$ with respect to modulo functions given in \hyperref[equ4]{eq. \ref{equ4}} and \hyperref[equ5]{eq. \ref{equ5}}, respectively. The set of quotient ring of the Hurwitz integers that denoted by $ \mathcal{H}_{\alpha}$ is shown by
\begin{equation}\label{equ7}
\mathcal{H}_{\alpha}= \{ \mu_{\alpha}(z) =  \mathrm{min} \{ \mu_{\alpha}^{(1)}(z), \mu_{\alpha}^{(2)}(z) \} | z \in \mathbb{Z}_{N(\pi)} \}.
\end{equation}
\end{definition}

\begin{proposition} \label{prop1}
Let $\alpha$ be a prime Hurwitz integer, and $z=0.$ Then
\begin{itemize}
  \item[i.] $\mu_{\alpha}^{(1)}(0)=0,$
  \item[ii.] $\mu_{\alpha}^{(2)}(0) \equiv 0 \mod \alpha.$
\end{itemize}
\end{proposition}
\begin{proof} Let $\alpha=\alpha_{1}+\alpha_{2}i+\alpha_{3}j+\alpha_{4}k$ be a prime Hurwitz integer,  and $z=0.$

$(i.)$ The proof can be easily seen from \hyperref[equ4]{eq. \ref{equ4}}.

$(ii.)$ From \hyperref[equ4]{eq. \ref{equ5}},
\begin{equation}\label{equu8}
\begin{array}{lll}
  \mu_{\alpha}^{(2)}(0) & = & 0 - \alpha \lfloor \lfloor \frac{\overline{\alpha} 0}{N(\alpha)} \rceil \rceil \\
   & = & - (\alpha_{1}+\alpha_{2}i+\alpha_{3}j+\alpha_{4}k) \lfloor \lfloor \frac{(\alpha_{1}-\alpha_{2}i-\alpha_{3}j-\alpha_{4}k) 0}{N(\alpha)} \rceil \rceil \\
   & = & - (\alpha_{1}+\alpha_{2}i+\alpha_{3}j+\alpha_{4}k) \lfloor \lfloor 0 \rceil \rceil \\
   & = & - (\alpha_{1}+\alpha_{2}i+\alpha_{3}j+\alpha_{4}k)(\frac{1}{2} + \frac{1}{2}i + \frac{1}{2}j + \frac{1}{2}k) \\
   & = & -\frac{\alpha_{1}}{2} - \frac{\alpha_{1}}{2}i - \frac{\alpha_{1}}{2}j - \frac{\alpha_{1}}{2}k - \frac{\alpha_{2}}{2}i + \frac{\alpha_{2}}{2} - \frac{\alpha_{2}}{2}k + \frac{\alpha_{2}}{2}j \\
   &  & - \frac{\alpha_{3}}{2}j + \frac{\alpha_{3}}{2}k + \frac{\alpha_{3}}{2} - \frac{\alpha_{3}}{2}i - \frac{\alpha_{4}}{2}k - \frac{\alpha_{4}}{2}j + \frac{\alpha_{4}}{2}i + \frac{\alpha_{4}}{2} \\
   & = & \frac{-\alpha_{1}+\alpha_{2}+\alpha_{3}+\alpha_{4}}{2} + \frac{-\alpha_{1} - \alpha_{2} - \alpha_{3} + \alpha_{4}}{2}i + \frac{-\alpha_{1} + \alpha_{2} - \alpha_{3} - \alpha_{4}}{2}j + \frac{-\alpha_{1} - \alpha_{2} + \alpha_{3} - \alpha_{4}}{2}k \\
   & \equiv & 0 \mod \alpha.
\end{array}
\end{equation}
This completes the proof.
\end{proof}

The set $\mathcal{H}_{\alpha} $ contains $N(\alpha)$ elements. The modulo function $\mu$ in \hyperref[def9]{definition \ref{def9}} defines a bijective mapping from $\mathbb{Z}_{N(\alpha)} $ into $\mathcal{H}_{\alpha}.$ In other words, there is a ring isomorphism between $\mathbb{Z}_{N(\alpha)} $ and $\mathcal{H}_{\alpha}.$ With the following theorems, we first show that the modulo function $\mu$ in \hyperref[def9]{definition \ref{def9}} is a ring homomorphism, and then we show that it is a ring isomorphism.

\begin{theorem}\label{thm1}
Let $\alpha$ be a prime Hurwitz integer, and $ z_{1},z_{2} \in \mathbb{Z}_{N(\alpha)}.$ The modulo function $\mu: \mathbb{Z}_{N(\alpha)} \rightarrow \mathcal{H}_{\alpha}$ is a ring homomorphism.
\end{theorem}

\begin{proof}
Let $\alpha$ be a prime Hurwitz integer. The modulo function $ \mu: \mathbb{Z}_{N(\alpha)} \rightarrow \mathcal{H}_{\alpha} $ by
\begin{equation}\label{equ8}
\mu_{\alpha}(z) = \mathrm{min} \{ \mu_{\alpha}^{(1)}(z), \mu_{\alpha}^{(2)}(z) \}
\end{equation}
where
\begin{equation}\label{equ9}
\mu_{\alpha}^{(1)}(z) = z \mod \alpha = z - \alpha \lfloor \frac{\overline{\alpha} z}{N(\alpha)} \rceil
\end{equation} and
\begin{equation}\label{equ10}
\mu_{\alpha}^{(2)}(z) = z \mod \alpha = z - \alpha \lfloor \lfloor \frac{\overline{\alpha} z}{N(\alpha)} \rceil \rceil
\end{equation} such that
\begin{equation}\label{equ11}
    \mathrm{min} \{ \mu_{\alpha}^{(1)}(z), \mu_{\alpha}^{(2)}(z) \} = \begin{cases}
                                                                  \mu_{\alpha}^{(1)}(z), & \mbox{if } N(\mu_{\alpha}^{(1)}(z)) < N(\mu_{\alpha}^{(2)}(z)) \\
                                                                  \mu_{\alpha}^{(1)}(z), & \mbox{if } N(\mu_{\alpha}^{(1)}(z)) = N(\mu_{\alpha}^{(2)}(z)) \mbox{ where z is an even integer} \\
                                                                  \mu_{\alpha}^{(2)}(z), & \mbox{otherwise}.
                                                                \end{cases}
\end{equation}
where $z \in \mathbb{Z}_{N(\alpha)}.$ From \hyperref[equ8]{eq. \ref{equ8}}, $\mu_{\alpha}(z_{1}) = \mathrm{min} \{ \mu_{\alpha}^{(1)}(z_{1}), \mu_{\alpha}^{(2)}(z_{1}) \},$ and $\mu_{\alpha}(z_{2}) = \mathrm{min} \{ \mu_{\alpha}^{(1)}(z_{2}), \mu_{\alpha}^{(2)}(z_{2}) \}$ where $z_{1},z_{2} \in \mathbb{Z}_{N(\alpha)}.$ If $z_{1},z_{2} \in \mathbb{Z}_{N(\alpha)},$ then there are four probable case with respect to \hyperref[equ8]{eq. \ref{equ8}};
\begin{itemize}
  \item[i.] $\mu_{\alpha}(z_{1})=\mu^{(1)}_{\alpha}(z_{1})$ and $\mu_{\alpha}(z_{2})=\mu^{(1)}_{\alpha}(z_{2})$ or,
  \item[ii.] $\mu_{\alpha}(z_{1})=\mu^{(1)}_{\alpha}(z_{1})$ and $\mu_{\alpha}(z_{2})=\mu^{(2)}_{\alpha}(z_{2})$ or,
  \item[iii.] $\mu_{\alpha}(z_{1})=\mu^{(2)}_{\alpha}(z_{1})$ and $\mu_{\alpha}(z_{2})=\mu^{(1)}_{\alpha}(z_{2})$ or,
  \item[iv.] $\mu_{\alpha}(z_{1})=\mu^{(2)}_{\alpha}(z_{1})$ and $\mu_{\alpha}(z_{2})=\mu^{(2)}_{\alpha}(z_{2}).$
\end{itemize}
Let us show that $\mu_{\alpha}(z_{1}+z_{2}) = \mu_{\alpha}(z_{1}) + \mu_{\alpha}(z_{2}).$

$(i.)$ Let $\mu_{\alpha}(z_{1})=\mu^{(1)}_{\alpha}(z_{1})$ and $\mu_{\alpha}(z_{2})=\mu^{(1)}_{\alpha}(z_{2}).$ Then $\mu_{\alpha}(z_{1}+z_{2})=\mu^{(1)}_{\alpha}(z_{1}) + \mu^{(1)}_{\alpha}(z_{2}).$  From \hyperref[equ9]{eq. \ref{equ9}},
\begin{equation}\label{equ12}
\begin{array}{lll}
  \mu^{(1)}_{\alpha}(z_{1}) + \mu^{(1)}_{\alpha}(z_{2}) & = & z_{1} - \alpha \lfloor \frac{\overline \alpha z_{1} }{N(\alpha)} \rceil + z_{2} - \alpha \lfloor \frac{\overline \alpha z_{2} }{N(\alpha)}\rceil \\
   & = & z_{1} + z_{2} - \alpha ( \lfloor \frac{\overline \alpha z_{1} }{N(\alpha)} \rceil + \lfloor \frac{\overline \alpha z_{2} }{N(\alpha)}\rceil ).
\end{array}
\end{equation}
There exist $\lambda_{1},\lambda_{2} \in \mathcal{H}$ such that $\lfloor \frac{\overline \alpha z_{1} }{N(\alpha)}\rceil=\lambda_{1}$ and  $\lfloor \frac{\overline \alpha z_{2} }{N(\alpha)}\rceil=\lambda_{2}.$ Hereby,
\begin{equation}\label{equ13}
\begin{array}{lll}
  \mu^{(1)}_{\alpha}(z_{1}) + \mu^{(1)}_{\alpha}(z_{2}) & = & z_{1} + z_{2} - \alpha (  \lambda_{1} + \lambda_{2} ).
\end{array}
\end{equation}
Let $\lambda_{1} + \lambda_{2} =\lambda$ where $\lambda \in \mathcal{H}.$ Thereby,
\begin{equation}\label{equ14}
\begin{array}{lll}
  \mu^{(1)}_{\alpha}(z_{1}) + \mu^{(1)}_{\alpha}(z_{2}) & = & z_{1} + z_{2} - \alpha \lambda.
\end{array}
\end{equation}
Since $\mu_{\alpha}(z_{1}+z_{2}) = (z_{1} + z_{2}) \mod \alpha,$ then there exists $\exists \beta \in \mathcal{H}$ such that $ \mu_{\alpha}(z_{1}+z_{2}) = z_{1} + z_{2} - \alpha \beta .$ So, we have
\begin{equation}\label{equ15}
\begin{array}{lll}
\mu^{(1)}_{\alpha}(z_{1}) + \mu^{(1)}_{\alpha}(z_{2}) & = & \mu_{\alpha}(z_{1}+z_{2}).
\end{array}
\end{equation}

$(ii.)$ Let $\mu_{\alpha}(z_{1})=\mu^{(1)}_{\alpha}(z_{1})$ and $\mu_{\alpha}(z_{2})=\mu^{(2)}_{\alpha}(z_{2}).$ Then $\mu_{\alpha}(z_{1}+z_{2})=\mu^{(1)}_{\alpha}(z_{1}) + \mu^{(2)}_{\alpha}(z_{2}).$  From \hyperref[equ9]{eq. \ref{equ9}}, and \hyperref[equ10]{eq. \ref{equ10}},
\begin{equation}\label{equ16}
\begin{array}{lll}
  \mu^{(1)}_{\alpha}(z_{1}) + \mu^{(2)}_{\alpha}(z_{2}) & = & z_{1} - \alpha \lfloor \frac{\overline \alpha z_{1} }{N(\alpha)} \rceil + z_{2} - \alpha \lfloor \lfloor \frac{\overline \alpha z_{2} }{N(\alpha)}\rceil \rceil \\
   & = & z_{1} + z_{2} - \alpha ( \lfloor \frac{\overline \alpha z_{1} }{N(\alpha)} \rceil + \lfloor\lfloor \frac{\overline \alpha z_{2} }{N(\alpha)}\rceil \rceil ).
\end{array}
\end{equation}
There exist $\lambda_{1},\lambda_{2} \in \mathcal{H}$ such that $\lfloor \frac{\overline \alpha z_{1} }{N(\alpha)}\rceil=\lambda_{1}$ and  $\lfloor \lfloor \frac{\overline \alpha z_{2} }{N(\alpha)}\rceil \rceil = \lambda_{2}.$ Hereby,
\begin{equation}\label{equ17}
\begin{array}{lll}
  \mu^{(1)}_{\alpha}(z_{1}) + \mu^{(2)}_{\alpha}(z_{2}) & = & z_{1} + z_{2} - \alpha (  \lambda_{1} + \lambda_{2} ).
\end{array}
\end{equation}
Let $\lambda_{1} + \lambda_{2} =\lambda$ where $\lambda \in \mathcal{H}.$ Thereby,
\begin{equation}\label{equ18}
\begin{array}{lll}
  \mu^{(1)}_{\alpha}(z_{1}) + \mu^{(2)}_{\alpha}(z_{2}) & = & z_{1} + z_{2} - \alpha \lambda.
\end{array}
\end{equation}
Since $\mu_{\alpha}(z_{1}+z_{2}) = (z_{1} + z_{2}) \mod \alpha,$ then there exists $\exists \beta \in \mathcal{H}$ such that $ \mu_{\alpha}(z_{1}+z_{2}) = z_{1} + z_{2} - \alpha \beta.$ So, we have
\begin{equation}\label{equ19}
\begin{array}{lll}
\mu^{(1)}_{\alpha}(z_{1}) + \mu^{(2)}_{\alpha}(z_{2}) & = & \mu_{\alpha}(z_{1}+z_{2}).
\end{array}
\end{equation}
Similarly, we can also show cases in the $(iii.)$ and $(iv.).$ Consequently,
\begin{equation}\label{}
\mu_{\alpha}(z_{1}+z_{2}) = \mu_{\alpha}(z_{1}) + \mu_{\alpha}(z_{2}).
\end{equation}
On the other hand, let us show that $\mu_{\alpha}(z_{1}z_{2}) = \mu_{\alpha}(z_{1})\mu_{\alpha}(z_{2}).$

$(i.)$ Let $\mu_{\alpha}(z_{1})=\mu^{(1)}_{\alpha}(z_{1})$ and $\mu_{\alpha}(z_{2})=\mu^{(1)}_{\alpha}(z_{2}).$ Then $\mu_{\alpha}(z_{1}z_{2})=\mu^{(1)}_{\alpha}(z_{1}) \mu^{(1)}_{\alpha}(z_{2}).$  From \hyperref[equ9]{eq. \ref{equ9}},
\begin{equation}\label{equ20}
\begin{array}{lll}
  \mu^{(1)}_{\alpha}(z_{1}) \mu^{(1)}_{\alpha}(z_{2}) & = & (z_{1} - \alpha \lfloor \frac{\overline \alpha z_{1} }{N(\alpha)} \rceil) ( z_{2} - \alpha \lfloor \frac{\overline \alpha z_{2} }{N(\alpha)}\rceil ) \\
   & = & z_{1} z_{2} - z_{1} \alpha \lfloor \frac{\overline \alpha z_{2} }{N(\alpha)}\rceil - z_{2}  \alpha \lfloor \frac{\overline \alpha z_{1} }{N(\alpha)} \rceil + \alpha \lfloor \frac{\overline \alpha z_{1} }{N(\alpha)} \rceil \alpha \lfloor \frac{\overline \alpha z_{2} }{N(\alpha)}\rceil.
\end{array}
\end{equation}
There exist $\lambda_{1},\lambda_{2} \in \mathcal{H}$ such that $\lfloor \frac{\overline \alpha z_{1} }{N(\alpha)}\rceil=\lambda_{1}$ and  $\lfloor \frac{\overline \alpha z_{2} }{N(\alpha)}\rceil=\lambda_{2}.$ Hereby,
\begin{equation}\label{equ21}
\begin{array}{lll}
  \mu^{(1)}_{\alpha}(z_{1}) \mu^{(1)}_{\alpha}(z_{2}) & = & z_{1} z_{2} - z_{1} \alpha \lambda_{2} - z_{2} \alpha \lambda_{1} + \alpha \lambda_{1} \alpha \lambda_{2}\\
                                                      & = & z_{1} z_{2} - \alpha (z_{1} \lambda_{2} + z_{2} \lambda_{1} + \lambda_{1} \alpha \lambda_{2}).
\end{array}
\end{equation}
Let $z_{1} \lambda_{2} + z_{2} \lambda_{1} + \lambda_{1} \alpha \lambda_{2} =\lambda$ where $\lambda \in \mathcal{H}.$ Thereby,
\begin{equation}\label{equ22}
\begin{array}{lll}
  \mu^{(1)}_{\alpha}(z_{1}) \mu^{(1)}_{\alpha}(z_{2}) & = & z_{1} z_{2} - \alpha \lambda.
\end{array}
\end{equation}
Since $\mu_{\alpha}(z_{1}z_{2}) = (z_{1} z_{2}) \mod \alpha,$ then there exists $ \exists \beta \in \mathcal{H}$ such that $ \mu_{\alpha}(z_{1}z_{2}) = z_{1} z_{2} - \alpha \beta.$ So, we have
\begin{equation}\label{equ23}
\begin{array}{lll}
\mu^{(1)}_{\alpha}(z_{1}) \mu^{(1)}_{\alpha}(z_{2}) & = & \mu_{\alpha}(z_{1}z_{2}).
\end{array}
\end{equation}
$(ii.)$ Let $\mu_{\alpha}(z_{1})=\mu^{(1)}_{\alpha}(z_{1})$ and $\mu_{\alpha}(z_{2})=\mu^{(2)}_{\alpha}(z_{2}).$ Then $\mu_{\alpha}(z_{1}z_{2})=\mu^{(1)}_{\alpha}(z_{1}) \mu^{(2)}_{\alpha}(z_{2}).$  From \hyperref[equ9]{eq. \ref{equ9}}, and \hyperref[equ10]{eq. \ref{equ10}},
\begin{equation}\label{equ25}
\begin{array}{lll}
  \mu^{(1)}_{\alpha}(z_{1}) \mu^{(2)}_{\alpha}(z_{2}) & = & (z_{1} - \alpha \lfloor \frac{\overline \alpha z_{1} }{N(\alpha)} \rceil) ( z_{2} - \alpha \lfloor \lfloor \frac{\overline \alpha z_{2} }{N(\alpha)}\rceil \rceil) \\
                                                      & = & z_{1} z_{2} - z_{1} \alpha \lfloor \lfloor \frac{\overline \alpha z_{2} }{N(\alpha)}\rceil \rceil - z_{2} \alpha \lfloor \frac{\overline \alpha z_{1} }{N(\alpha)} \rceil + \alpha \lfloor \frac{\overline \alpha z_{1} }{N(\alpha)} \rceil \alpha \lfloor \lfloor \frac{\overline \alpha z_{2} }{N(\alpha)}\rceil \rceil.
\end{array}
\end{equation}
There exist $\lambda_{1},\lambda_{2} \in \mathcal{H}$ such that $\lfloor \frac{\overline \alpha z_{1} }{N(\alpha)}\rceil=\lambda_{1}$ and  $\lfloor \lfloor \frac{\overline \alpha z_{2} }{N(\alpha)}\rceil \rceil = \lambda_{2}.$ Hereby,
\begin{equation}\label{equ26}
\begin{array}{lll}
  \mu^{(1)}_{\alpha}(z_{1}) \mu^{(2)}_{\alpha}(z_{2}) & = & z_{1} z_{2} - z_{1} \alpha \lambda_{2} - z_{2} \alpha \lambda_{1} + \alpha \lambda_{1} \alpha \lambda_{2} \\
                                                      & = & z_{1} z_{2} - \alpha (z_{1} \lambda_{2} + z_{2} \lambda_{1} + \lambda_{1} \alpha \lambda_{2}).
\end{array}
\end{equation}
Let $z_{1} \lambda_{2} + z_{2} \lambda_{1} + \lambda_{1} \alpha \lambda_{2} =\lambda$ where $\lambda \in \mathcal{H}.$ Thereby,
\begin{equation}\label{equ27}
\begin{array}{lll}
  \mu^{(1)}_{\alpha}(z_{1}) \mu^{(2)}_{\alpha}(z_{2}) & = & z_{1} z_{2} - \alpha \lambda.
\end{array}
\end{equation}
Since $\mu_{\alpha}(z_{1}z_{2}) = (z_{1} z_{2}) \mod \alpha,$ then there exists $\exists \beta \in \mathcal{H}$ such that $ \mu_{\alpha}(z_{1}z_{2}) = z_{1} z_{2} - \alpha \beta.$ So, we have
\begin{equation}\label{equ28}
\begin{array}{lll}
\mu^{(1)}_{\alpha}(z_{1}) \mu^{(2)}_{\alpha}(z_{2}) & = & \mu_{\alpha}(z_{1}z_{2}).
\end{array}
\end{equation}
Similarly, we can also show cases in the $(iii.)$ and $(iv.).$ Consequently,
\begin{equation}\label{}
\mu_{\alpha}(z_{1}z_{2}) = \mu_{\alpha}(z_{1}) \mu_{\alpha}(z_{2}).
\end{equation}
Namely, the modulo function $\mu: \mathbb{Z}_{N(\alpha)} \rightarrow \mathcal{H}_{\alpha}$ is a ring homomorphism. This completes this proof.
\end{proof}

\begin{theorem}\label{thm2}
The modulo function $\mu: \mathbb{Z}_{N(\alpha)} \rightarrow \mathcal{H}_{\alpha}$ is a ring isomorphism ring. Namely, $\mathbb{Z}_{N(\alpha)} \cong \mathcal{H}_{\alpha}.$
\end{theorem}

\begin{proof}
Let $\alpha$ be a prime Hurwitz integer and, $z_{1},z_{2} \in \mathbb{Z}_{N(\alpha)}.$ The modulo function $\mu : \mathbb{Z}_{N(\alpha)} \rightarrow \mathcal{H}_\alpha :$
\begin{equation}\label{equ29}
\mu_{\alpha}(z) = \mathrm{min} \{ \mu_{\alpha}^{(1)}(z), \mu_{\alpha}^{(2)}(z) \}
\end{equation}
where
\begin{equation}\label{equ30}
\mu_{\alpha}^{(1)}(z) = z \mod \alpha = z - \alpha \lfloor \frac{\overline{\alpha} z}{N(\alpha)} \rceil
\end{equation} and
\begin{equation}\label{equ31}
\mu_{\alpha}^{(2)}(z) = z \mod \alpha = z - \alpha \lfloor \lfloor \frac{\overline{\alpha} z}{N(\alpha)} \rceil \rceil
\end{equation} such that
\begin{equation}\label{equ32}
    \mathrm{min} \{ \mu_{\alpha}^{(1)}(z), \mu_{\alpha}^{(2)}(z) \} = \begin{cases}
                                                                  \mu_{\alpha}^{(1)}(z), & \mbox{if } N(\mu_{\alpha}^{(1)}(z)) < N(\mu_{\alpha}^{(2)}(z)) \\
                                                                  \mu_{\alpha}^{(1)}(z), & \mbox{if } N(\mu_{\alpha}^{(1)}(z)) = N(\mu_{\alpha}^{(2)}(z)) \mbox{ where z is an even integer} \\
                                                                  \mu_{\alpha}^{(2)}(z), & \mbox{otherwise}.
                                                                \end{cases}
\end{equation}
where $z \in \mathbb{Z}_{N(\alpha)}.$

According to \hyperref[thm1]{theorem \ref{thm1}}, the modulo function $\mu$ is a ring homomorphism. We should show that it is a bijective ring homomorphism, i.e., a ring isomorphism. This mapping is a surjective ring homomorphism because of $\text{Im}\mu= \{ \mu_{\alpha}(z) = \mathrm{min} \{ \mu_{\alpha}^{(1)}(z), \mu_{\alpha}^{(2)}(z) \} : z \in \mathbb{Z}_{N(\alpha)} \} = \mathcal{H}_{\alpha}.$ If $z=0,$ then
\begin{equation}\label{equ33}
\begin{array}{lll}
  \mu_{\alpha}(0) & = &  \mathrm{min} \{ \mu_{\alpha}^{(1)}(0), \mu_{\alpha}^{(2)}(0) \}.
\end{array}
\end{equation}
From \hyperref[prop1]{proposition \ref{prop1}},
\begin{equation}\label{equ34}
\begin{array}{lll}
  \mu_{\alpha}(0) & = & 0.
\end{array}
\end{equation}
If $z\neq0,$ then $\mu_{\alpha}(z)$ is to greater than or equal to $1.$ Hereby, this mapping is an injective ring homomorphism because of $ Ker\mu =  \{ z\in\mathbb{Z}_{N(\alpha)} : \mu_{\alpha}(z)=0 \}= \{z\in\mathbb{Z}_{N(\alpha)} : z=0\}=\{0\}.$ So, $\mu$ function is a ring isomorphism since it is both a surjective ring homomorphism and an injective ring homomorphism, i.e. ${\mathbb{Z}_{N(\alpha)}} \cong {\mathcal{H}_\alpha}.$ This completes the proof.
\end{proof}

\begin{definition} \label{def11}
Let $\mathbb{F}_{q}$ be a finite field with $q$ elements. $C$ is called a code if $C$ is a nonempty subset of $\mathbb{F}_{q}^{n}.$ An element of $C$ is called a codeword in $C.$ A linear code $C$ of length $n$ over $\mathbb{F}_{q}$ is defined to be a subspace of $\mathbb{F}_{q}^{n}.$
\end{definition}

\begin{definition} \label{def12}
A code $C$ of length $n$ is a subset of the direct product $\mathcal{H}^{n}$ of $n$ copies of $\mathcal{H}.$ In each of the cases we consider, $\mathcal{H}$ is an Abelian group, and thus the same is true for $\mathcal{H}^{n}.$ A code $C$ is a group code if it is a subgroup of $\mathcal{H}^{n},$ or equivalently as $\mathcal{H}^{n}$ is a finite group,
\begin{equation}\nonumber
c,c^{'}\in C \Rightarrow c-c^{'}\in C.
\end{equation}
In the case when $\mathcal{H}$ is a finite field, and thus $\mathcal{H}^{n}$ is a vector space of dimension $n$ over $\mathcal{H},$ then a linear code is a subspace $C$ of $\mathcal{H}^{n}.$ Here we say that a code $C$ in $\mathcal{H}^{n}$ is an $(n,k)-code$ if the size of $C$ is equal to $|\mathcal{H}|^{k}$ (The case $k=0$ is of less interest, and thus left aside) \cite{bib12}.
\end{definition}

\begin{definition} \label{def13}
The average energy of $\mathcal{H}_{\alpha}$ that denoted by $\mathcal{E}_{\pi}$ is computed by
\begin{equation}\label{equ127}
\mathcal{E}_{\pi} = \frac{1}{N(\alpha)} \sum_{z=0}^{N(\alpha)-1}N(\mu_{\alpha}(z))
\end{equation}
where $z\in\mathbb{Z}_{N(\alpha)}$.
\end{definition}

\section{Some Results}\label{sec4}

This section presents some results related to the modulo function $\mu$ defined in \hyperref[def10]{definition \ref{def10}} and, the mathematical notations given in \hyperref[equ4]{eq. \ref{equ4}} and \hyperref[equ5]{eq. \ref{equ5}}. $\mathcal{H}^{(1)}_{\alpha}$ is the set of residual class with respect to \hyperref[equ4]{eq. \ref{equ4}}, $\mathcal{H}^{(2)}_{\alpha}$ is the set of residual class with respect to \hyperref[equ5]{eq. \ref{equ5}}, and $\mathcal{H}_{\alpha}$ is the set of residual class with respect to the modulo function $\mu$ defined in \hyperref[def10]{definition \ref{def10}}. $2+i,$ $3+2j,$ and $4i+k$ are prime Hurwitz integers that have two-components, and so on. $3+i+j,$ $3+2j+2k,$ and $4i+3j+2k$ are prime Hurwitz integers that have three-components, and so on.

\begin{corollary} \label{cor1}
Let $\alpha$ be a prime Hurwitz integer that have two-components. Then, we have
\begin{equation}\label{equu35}
\mathcal{H}_{\alpha}=\mathcal{H}^{(1)}_{\alpha}.
\end{equation}
\end{corollary}

\begin{proposition} \label{prop2}
Let $\beta=\beta_{1}+\beta_{2}i+\beta_{3}j+\beta_{4}k$ be a Hurwitz integer. If $\beta_{1}, \beta_{2}, \beta_{3}, \beta_{4} \in \mathbb{Z},$ then,
\begin{equation}\label{equ35}
\lfloor \beta \rceil = \beta,
\end{equation}
and
\begin{equation}\label{equ36}
\lfloor \lfloor \beta \rceil \rceil=\beta \pm \frac{1}{2} \pm \frac{1}{2}i \pm \frac{1}{2}j \pm \frac{1}{2}k.
\end{equation}
If $\beta_{1}, \beta_{2}, \beta_{3}, \beta_{4} \in \mathbb{Z}+\frac{1}{2},$ then,
\begin{equation}\label{equ37}
\lfloor \beta \rceil = \beta \pm \frac{1}{2} \pm \frac{1}{2}i \pm \frac{1}{2}j \pm \frac{1}{2}k,
\end{equation}
and
\begin{equation}\label{equ38}
\lfloor \lfloor \beta \rceil \rceil=\beta.
\end{equation}
\end{proposition}

\begin{proof}
Let $\beta=\beta_{1}+\beta_{2}i+\beta_{3}j+\beta_{4}k$ be a Hurwitz integer. If $\beta_{1}, \beta_{2}, \beta_{3}, \beta_{4} \in \mathbb{Z},$ then,
\begin{equation}\label{equ39}
           \begin{array}{lll}
             \lfloor \beta \rceil & = & \lfloor \beta_{1}+\beta_{2}i+\beta_{3}j+\beta_{4}k \rceil \\
                                  & = & \lfloor \beta_{1}\rceil + \lfloor \beta_{2} \rceil i + \lfloor \beta_{3} \rceil j + \lfloor \beta_{4}\rceil k \\
                                  & = & \beta_{1}+\beta_{2}i+\beta_{3}j+\beta_{4}k \\
                                  & = & \beta.
           \end{array}
\end{equation}
On the other hand,
\begin{equation}\label{equ41}
           \begin{array}{lll}
             \lfloor \lfloor \beta \rceil \rceil & = & \lfloor \lfloor \beta_{1}+\beta_{2}i+\beta_{3}j+\beta_{4}k \rceil \rceil \\
                                  & = & \lfloor \lfloor \beta_{1} \rceil \rceil + \lfloor \lfloor \beta_{2} \rceil \rceil i + \lfloor \lfloor \beta_{3} \rceil \rceil j + \lfloor \lfloor \beta_{4} \rceil \rceil k.
           \end{array}
           \end{equation}
Since the property of rounding notation, then,
           \begin{equation}\label{equ42}
           \begin{array}{lll}
             \lfloor \lfloor \beta \rceil \rceil & = & \beta_{1} \pm \frac{1}{2} + (\beta_{2} \pm \frac{1}{2})i + (\beta_{3} \pm \frac{1}{2})j + (\beta_{4} \pm \frac{1}{2})k \\
                                  & = & \beta_{1} + \beta_{2}i + \beta_{3}j + \beta_{4}k \pm \frac{1}{2} \pm \frac{1}{2}i \pm \frac{1}{2}j \pm \frac{1}{2}k \\
                                  & = & \beta \pm \frac{1}{2} \pm \frac{1}{2}i \pm \frac{1}{2}j \pm \frac{1}{2}k \\
           \end{array}
           \end{equation}
If $\beta_{1}, \beta_{2}, \beta_{3}, \beta_{4} \in \mathbb{Z}+\frac{1}{2},$ then,
\begin{equation}\label{equ43}
           \begin{array}{lll}
             \lfloor \beta \rceil & = & \lfloor \beta_{1}+\beta_{2}i+\beta_{3}j+\beta_{4}k \rceil \\
                                  & = & \lfloor \beta_{1} \rceil + \lfloor \beta_{2} \rceil i + \lfloor \beta_{3} \rceil j + \lfloor \beta_{4}\rceil k .
           \end{array}
        \end{equation}
Since the property of rounding notation, then,
        \begin{equation}\label{equ44}
        \begin{array}{lll}
             \begin{array}{lll}
             \lfloor \beta \rceil & = & \beta_{1} \pm \frac{1}{2} + (\beta_{2} \pm \frac{1}{2})i + (\beta_{3} \pm \frac{1}{2})j + (\beta_{4} \pm \frac{1}{2})k \\
                                  & = & \beta_{1} + \beta_{2}i + \beta_{3}j + \beta_{4}k \pm \frac{1}{2} \pm \frac{1}{2}i \pm \frac{1}{2}j \pm \frac{1}{2}k \\
                                  & = & \beta \pm \frac{1}{2} \pm \frac{1}{2}i \pm \frac{1}{2}j \pm \frac{1}{2}k.
           \end{array}
           \end{array}
        \end{equation} \\
On the other hand,
\begin{equation}\label{equ44}
           \begin{array}{lll}
             \lfloor \lfloor \beta \rceil \rceil & = & \lfloor \lfloor \beta_{1}+\beta_{2}i+\beta_{3}j+\beta_{4}k \rceil \rceil \\
                                  & = & \lfloor \lfloor \beta_{1} \rceil \rceil + \lfloor \lfloor \beta_{2} \rceil \rceil i + \lfloor \lfloor \beta_{3} \rceil \rceil j + \lfloor \lfloor \beta_{4} \rceil \rceil k \\
                                  & = & \beta_{1} + \beta_{2}i + \beta_{3}j + \beta_{4}k \\
                                  & = & \beta.
           \end{array}
           \end{equation}
This completes the proof.
\end{proof}

\begin{corollary}\label{prop3}
Let $\beta=\beta_{1}+\beta_{2}i+\beta_{3}j+\beta_{4}k$ and $\pi = \pi_{1} + \pi_{2}i + \pi_{3}j + \pi_{4}k$ be Hurwitz integers. If $\beta_{1}, \beta_{2}, \beta_{3}, \beta_{4} \in \mathbb{Z},$ and $\pi_{1}, \pi_{2}, \pi_{3}, \pi_{4} \in \mathbb{Z},$ then,
\begin{equation}\label{equ46}
\lfloor \pi \pm \beta \rceil = \lfloor \pi \rceil \pm \lfloor \beta \rceil = \pi \pm \beta,
\end{equation}
and
\begin{equation}\label{equ47}
\lfloor \lfloor \pi \pm \beta \rceil \rceil = \left( \pi \pm \frac{1}{2} \pm \frac{1}{2}i \pm \frac{1}{2}j \pm \frac{1}{2}k  \right) \pm \left( \beta  \pm \frac{1}{2} \pm \frac{1}{2}i \pm \frac{1}{2}j \pm \frac{1}{2}k \right).
\end{equation}
If $\beta_{1}, \beta_{2}, \beta_{3}, \beta_{4} \in \mathbb{Z},$ and $\pi_{1}, \pi_{2}, \pi_{3}, \pi_{4} \in \mathbb{Z}+\frac{1}{2},$ then,
\begin{equation}\label{equ48}
\lfloor \pi \pm \beta \rceil = \lfloor \pi \rceil \pm \lfloor \beta \rceil = \pi \pm \beta \pm \frac{1}{2} \pm \frac{1}{2}i \pm \frac{1}{2}j \pm \frac{1}{2}k,
\end{equation}
and
\begin{equation}\label{equ49}
\lfloor \lfloor \pi \pm \beta \rceil \rceil = \lfloor\lfloor \pi \rceil\rceil \pm \lfloor\lfloor \beta \rceil\rceil = \pi \pm \beta \pm \frac{1}{2} \pm \frac{1}{2}i \pm \frac{1}{2}j \pm \frac{1}{2}k.
\end{equation}
If $\beta_{1}, \beta_{2}, \beta_{3}, \beta_{4} \in \mathbb{Z}+\frac{1}{2},$ and $\pi_{1}, \pi_{2}, \pi_{3}, \pi_{4} \in \mathbb{Z}+\frac{1}{2},$ then,
\begin{equation}\label{equ50}
\lfloor \pi \pm \beta \rceil = \lfloor \pi \rceil \pm \lfloor \beta \rceil = \left( \pi \pm \frac{1}{2} \pm \frac{1}{2}i \pm \frac{1}{2}j \pm \frac{1}{2}k  \right) \pm \left( \beta  \pm \frac{1}{2} \pm \frac{1}{2}i \pm \frac{1}{2}j \pm \frac{1}{2}k \right),
\end{equation}
and
\begin{equation}\label{equ51}
\lfloor \lfloor \pi \pm \beta \rceil \rceil = \lfloor\lfloor \pi \rceil\rceil \pm \lfloor\lfloor \beta \rceil\rceil = \pi \pm \beta.
\end{equation}
\end{corollary}

\begin{proof}
The proof can be easily seen from \hyperref[prop2]{proposition \ref{prop2}}.
\end{proof}

\begin{proposition} \label{prop4}
Let $\beta=\beta_{1}+\beta_{2}i+\beta_{3}j+\beta_{4}k$ and $\pi = \pi_{1} + \pi_{2}i + \pi_{3}j + \pi_{4}k$ be Hurwitz integers. We consider $\pi_{1}$ and $\beta_{1}.$ If $\pi_{1}, \beta_{1} \in \mathbb{Z},$ then,
\begin{equation}\label{equ52}
\lfloor \pi_{1} \beta_{1} \rceil = \pi_{1} \beta_{1},
\end{equation}
and
\begin{equation}\label{equ53}
\lfloor \lfloor \pi_{1} \beta_{1} \rceil \rceil= \pi_{1} \beta_{1} \pm \frac{1}{2}.
\end{equation}
Let $\pi_{1} \in \mathbb{Z}$ and $\beta_{1} \in \mathbb{Z}+\frac{1}{2}.$ If $\pi_{1}$ is an even integer,  then,
\begin{equation}\label{equ54}
\lfloor \pi_{1} \beta_{1} \rceil = \pi_{1} \beta_{1},
\end{equation}
and
\begin{equation}\label{equ55}
\lfloor \lfloor \pi_{1} \beta_{1} \rceil \rceil= \pi_{1} \beta_{1} \pm \frac{1}{2}.
\end{equation}
If $\pi_{1}$ is an odd integer, then,
\begin{equation}\label{equ56}
\lfloor \pi_{1} \beta_{1} \rceil = \pi_{1} \beta_{1} \pm \frac{1}{2},
\end{equation}
and
\begin{equation}\label{equ57}
\lfloor \lfloor \pi_{1} \beta_{1} \rceil \rceil= \pi_{1} \beta_{1}.
\end{equation}
Let $\pi_{1}, \beta_{1} \in \mathbb{Z}+\frac{1}{2}$ such that $\pi=\frac{\lambda_{1}}{2}$ and $\beta=\frac{\lambda_{2}}{2}$ where $\lambda_{1}$ and $\lambda_{2}$ are odd integers. If $\lambda_{1}\lambda_{2} \equiv 1 \mod 4,$ then,
\begin{equation}\label{equ58}
\lfloor \pi_{1} \beta_{1} \rceil = \pm \pi_{1} \beta_{1} \mp \frac{1}{4},
\end{equation}
and
\begin{equation}\label{equ59}
\lfloor \lfloor \pi_{1} \beta_{1} \rceil \rceil= \pm \pi_{1} \beta_{1} \pm \frac{1}{4}.
\end{equation}
If $\lambda_{1}\lambda_{2} \equiv 3 \mod 4,$ then,
\begin{equation}\label{equ58}
\lfloor \pi_{1} \beta_{1} \rceil = \pm \pi_{1} \beta_{1} \pm \frac{1}{4},
\end{equation}
and
\begin{equation}\label{equ59}
\lfloor \lfloor \pi_{1} \beta_{1} \rceil \rceil= \pm \pi_{1} \beta_{1} \mp \frac{1}{4}.
\end{equation}
\end{proposition}

\begin{proof} Let $\beta=\beta_{1}+\beta_{2}i+\beta_{3}j+\beta_{4}k$ and $\pi = \pi_{1} + \pi_{2}i + \pi_{3}j + \pi_{4}k$ be Hurwitz integers. We consider $\pi_{1}$ and $\beta_{1}.$ Let $\pi_{1}, \beta_{1} \in \mathbb{Z}.$ $\pi_{1}\beta_{1} \in \mathbb{Z}$ because of $\pi_{1}, \beta_{1} \in \mathbb{Z}.$ So, from the property of round notations,
\begin{equation}\label{equ60}
\lfloor \pi_{1} \beta_{1} \rceil = \pi_{1} \beta_{1},
\end{equation}
and
\begin{equation}\label{equ61}
\lfloor \lfloor \pi_{1} \beta_{1} \rceil \rceil= \pi_{1} \beta_{1} \pm \frac{1}{2}.
\end{equation}
Let $\pi_{1} \in \mathbb{Z}$ and $\beta_{1} \in \mathbb{Z}+\frac{1}{2}.$ $\pi_{1}\beta_{1} \in \mathbb{Z}$ because of $\pi_{1} \in \mathbb{Z},$ $\beta_{1} \in \mathbb{Z}+\frac{1}{2},$ and $\pi_{1}$ is an even integer. So, from the property of round notations,
\begin{equation}\label{equ62}
\lfloor \pi_{1} \beta_{1} \rceil = \pi_{1} \beta_{1},
\end{equation}
and
\begin{equation}\label{equ63}
\lfloor \lfloor \pi_{1} \beta_{1} \rceil \rceil= \pi_{1} \beta_{1} \pm \frac{1}{2}.
\end{equation}
$\pi_{1}\beta_{1} \in \mathbb{Z}+\frac{1}{2}$ because of $\pi_{1} \in \mathbb{Z},$ $\beta_{1} \in \mathbb{Z}+\frac{1}{2},$ and $\pi_{1}$ is an odd integer. So, from the property of round notations,
\begin{equation}\label{equ64}
\lfloor \pi_{1} \beta_{1} \rceil = \pi_{1} \beta_{1} \pm \frac{1}{2},
\end{equation}
and
\begin{equation}\label{equ65}
\lfloor \lfloor \pi_{1} \beta_{1} \rceil \rceil= \pi_{1} \beta_{1}.
\end{equation}
Let $\pi_{1}, \beta_{1} \in \mathbb{Z}+\frac{1}{2}$ and, let $\pi_{1}=\frac{\lambda_{1}}{2}$ and $\beta_{1}=\frac{\lambda_{2}}{2}$ where $\lambda_{1}$ and $\lambda_{2}$ are odd integers. If $1 \equiv \lambda_{1} \lambda_{2} \mod 4,$ then, $\lambda_{1}\lambda_{2} = 4k +1$ where $k\in\mathbb{Z}.$ Hereby,
\begin{equation}\label{equ66}
\begin{array}{lll}
  \lfloor \pi_{1} \beta_{1} \rceil & = & \lfloor \pm \frac{\lambda_{1}\lambda_{2}}{4} \rceil \\
   & = & \lfloor \pm \frac{4k+1}{4} \rceil.
\end{array}
\end{equation}
$\lfloor \frac{4k+1}{4} \rceil = \pm k$ because of the property of round notation, and $ k < \lfloor \frac{4k+1}{4} \rceil < \frac{2k+1}{2}.$. Hereby,
\begin{equation}\label{equ67}
\begin{array}{lll}
  \lfloor \pi_{1} \beta_{1} \rceil & = & \pm k
\end{array}
\end{equation}
Since $k=\frac{\lambda_{1}\lambda_{2}}{4}-\frac{1}{4} = \pi_{1} \beta_{1}-\frac{1}{4},$ then
\begin{equation}\label{equ68}
\begin{array}{lll}
  \lfloor \pi_{1} \beta_{1} \rceil & = & \pm \pi_{1} \beta_{1} \mp \frac{1}{4}.
\end{array}
\end{equation}
On the other hand,
\begin{equation}\label{equ69}
\begin{array}{lll}
  \lfloor \lfloor \pi_{1} \beta_{1} \rceil \rceil & = & \lfloor \lfloor \pm \frac{\lambda_{1} \lambda_{2}}{4} \rceil \rceil \\
   & = & \lfloor \lfloor \pm \frac{4k+1}{4} \rceil \rceil \\
\end{array}
\end{equation}
$ \pm \lfloor \lfloor \frac{4k+1}{4} \rceil \rceil = \pm \frac{2k+1}{2}$ because of the property of round notation, and $k < \lfloor \lfloor \frac{4k+1}{4} \rceil \rceil < \frac{2k+1}{2}.$ Hereby,
\begin{equation}\label{equ70}
\begin{array}{lll}
  \lfloor \lfloor \pi_{1} \beta_{1} \rceil \rceil & = & \pm \frac{2k+1}{2}\\
  & = & \pm k \pm \frac{1}{2}.
\end{array}
\end{equation}
Since $k=\frac{\lambda_{1}\lambda_{2}}{4}-\frac{1}{4} = \pi_{1} \beta_{1}-\frac{1}{4},$ then
\begin{equation}\label{equ71}
\begin{array}{lll}
  \lfloor \lfloor \pi_{1} \beta_{1} \rceil \rceil & = & \pm \pi_{1} \beta_{1} \mp \frac{1}{4} \pm \frac{1}{2} \\
  & = & \pm \pi_{1} \beta_{1} \pm \frac{1}{4}.
\end{array}
\end{equation}
If $3 \equiv \lambda_{1} \lambda_{2} \mod 4,$ then, $\lambda_{1} \lambda_{2} = 4k + 3$ where $k\in\mathbb{Z}.$ Hereby,
\begin{equation}\label{equ72}
\begin{array}{lll}
  \lfloor \pi_{1} \beta_{1} \rceil & = & \lfloor \pm \frac{\lambda_{1} \lambda_{2}}{4} \rceil \\
   & = & \lfloor \pm \frac{4k+3}{4} \rceil.
\end{array}
\end{equation}
$\lfloor \frac{4k+3}{4} \rceil = \pm k \pm 1$ because of the property of round notation, and $\frac{2k+1}{2} < \lfloor \frac{4k+3}{4} \rceil < k+1.$ Hereby,
\begin{equation}\label{equ73}
\begin{array}{lll}
  \lfloor \pi_{1} \beta_{1} \rceil & = & \pm k \pm 1.
\end{array}
\end{equation}
Since $k=\frac{\lambda_{1}\lambda{2}}{4}-\frac{3}{4} = \pi_{1} \beta_{1} - \frac{3}{4},$ then
\begin{equation}\label{equ74}
\begin{array}{lll}
  \lfloor \pi_{1} \beta_{1} \rceil & = & \pm \pi_{1} \beta_{1} \mp \frac{3}{4} \pm 1 \\
  & = & \pm \pi_{1} \beta_{1} \pm \frac{1}{4}.
\end{array}
\end{equation}
On the other hand,
\begin{equation}\label{equ75}
\begin{array}{lll}
  \lfloor \lfloor \pi_{1} \beta_{1} \rceil \rceil & = & \lfloor \lfloor \pm \frac{\lambda_{1}}{4} \rceil \rceil \\
   & = & \lfloor \lfloor \pm \frac{4k+3}{4} \rceil \rceil \\
\end{array}
\end{equation}
$ \pm \lfloor \lfloor \frac{4k+3}{4} \rceil \rceil = \pm \frac{2k+1}{2}$ because of the property of round notation, and $ \frac{2k+1}{2} < \lfloor \lfloor \frac{4k+3}{4} \rceil \rceil < k + 1.$ Hereby,
\begin{equation}\label{equ76}
\begin{array}{lll}
  \lfloor \lfloor \pi_{1} \beta_{1} \rceil \rceil & = & \pm \frac{2k+1}{2}\\
  & = & \pm k \pm \frac{1}{2}.
\end{array}
\end{equation}
Since $k=\frac{\lambda_{1}\lambda_{2}}{4}-\frac{3}{4} = \pi_{1} \beta_{1}-\frac{3}{4},$ then
\begin{equation}\label{equ77}
\begin{array}{lll}
  \lfloor \lfloor \pi_{1} \beta_{1} \rceil \rceil & = & \pm \pi_{1} \beta_{1} \mp \frac{3}{4} \pm \frac{1}{2} \\
  & = & \pm \pi_{1} \beta_{1} \mp \frac{1}{4}.
\end{array}
\end{equation}
This completes the proof.
\end{proof}

\begin{proposition}\label{prop5}
Let $\alpha$ be a prime Hurwitz integer. Then,
\begin{equation}\label{equ78}
N(\mu_{\alpha}(0))=N(\alpha).
\end{equation}
\end{proposition}

\begin{proof}
Let $\alpha$ be a prime Hurwitz integer. From \hyperref[equu8]{eq. \ref{equu8}}, we know that
\begin{equation}\label{equ79}
\mu_{\alpha}(0)=\frac{-\alpha_{1}+\alpha_{2}+\alpha_{3}+\alpha_{4}}{2} + \frac{-\alpha_{1} - \alpha_{2} - \alpha_{3} + \alpha_{4}}{2}i + \frac{-\alpha_{1} + \alpha_{2} - \alpha_{3} - \alpha_{4}}{2}j + \frac{-\alpha_{1} - \alpha_{2} + \alpha_{3} - \alpha_{4}}{2}k.
\end{equation}
Hereby,
\begin{equation}\label{equ80}
\begin{array}{lll}
  N(\mu_{\alpha}(0)) & = & N(\frac{-\alpha_{1}+\alpha_{2}+\alpha_{3}+\alpha_{4}}{2} + \frac{-\alpha_{1} - \alpha_{2} - \alpha_{3} + \alpha_{4}}{2}i + \frac{-\alpha_{1} + \alpha_{2} - \alpha_{3} - \alpha_{4}}{2}j + \frac{-\alpha_{1} - \alpha_{2} + \alpha_{3} - \alpha_{4}}{2}k) \\
   & = & (\frac{-\alpha_{1}+\alpha_{2}+\alpha_{3}+\alpha_{4}}{2})^{2} + (\frac{-\alpha_{1} - \alpha_{2} - \alpha_{3} + \alpha_{4}}{2})^{2} + (\frac{-\alpha_{1} + \alpha_{2} - \alpha_{3} - \alpha_{4}}{2})^{2} + (\frac{-\alpha_{1} - \alpha_{2} + \alpha_{3} - \alpha_{4}}{2})^{2} \\
   & = & \frac{\alpha^{2}_{1} + \alpha^{2}_{2} + \alpha^{2}_{3} + \alpha^{2}_{4} + 2\alpha_{3}\alpha_{4} + 2\alpha_{2}\alpha_{3} + 2\alpha_{2}\alpha_{4} - 2\alpha_{1}\alpha_{2} - 2\alpha_{1}\alpha_{3} - 2\alpha_{1}\alpha_{4}}{4} \\
   &   & + \frac{\alpha^{2}_{1} + \alpha^{2}_{2} + \alpha^{2}_{3} + \alpha^{2}_{4} - 2\alpha_{3}\alpha_{4} + 2\alpha_{2}\alpha_{3} - 2\alpha_{2}\alpha_{4} + 2\alpha_{1}\alpha_{2} + 2\alpha_{1}\alpha_{3} - 2\alpha_{1}\alpha_{4}}{4} \\
   &   & + \frac{\alpha^{2}_{1} + \alpha^{2}_{2} + \alpha^{2}_{3} + \alpha^{2}_{4} + 2\alpha_{3}\alpha_{4} - 2\alpha_{2}\alpha_{3} - 2\alpha_{2}\alpha_{4} - 2\alpha_{1}\alpha_{2} + 2\alpha_{1}\alpha_{3} + 2\alpha_{1}\alpha_{4}}{4} \\
   &   & + \frac{\alpha^{2}_{1} + \alpha^{2}_{2} + \alpha^{2}_{3} + \alpha^{2}_{4} - 2\alpha_{3}\alpha_{4} - 2\alpha_{2}\alpha_{3} + 2\alpha_{2}\alpha_{4} + 2\alpha_{1}\alpha_{2} - 2\alpha_{1}\alpha_{3} + 2\alpha_{1}\alpha_{4}}{4} \\
   & =  & \frac{4(\alpha^{2}_{1} + \alpha^{2}_{2} + \alpha^{2}_{3} + \alpha^{2}_{4})}{4} \\
\end{array}
\end{equation}
Since $N(\alpha)=\alpha^{2}_{1} + \alpha^{2}_{2} + \alpha^{2}_{3} + \alpha^{2}_{4},$ then,
\begin{equation}\label{equ81}
\begin{array}{lll}
  N(\mu_{\alpha}(0)) & = & \alpha^{2}_{1} + \alpha^{2}_{2} + \alpha^{2}_{3} + \alpha^{2}_{4} \\
  & = & N(\alpha).
\end{array}
\end{equation}
This completes the proof.
\end{proof}

\begin{proposition} \label{prop6}
Let $\alpha=\alpha_{1}+\alpha_{2}i+\alpha_{3}j+\alpha_{4}k$ be a prime Hurwitz integer. If $\alpha_{1},\alpha_{2},\alpha_{3},\alpha_{4} \in \mathbb{Z},$ then,
\begin{equation}\label{equ82}
\mu^{(1)}_{\alpha}(N(\alpha)) = 0,
\end{equation}
and
\begin{equation}\label{equ83}
\mu^{(2)}_{\alpha}(N(\alpha)) \equiv 0 \mod \alpha.
\end{equation}
If $\alpha_{1},\alpha_{2},\alpha_{3},\alpha_{4} \in \mathbb{Z}+\frac{1}{2},$ then
\begin{equation}\label{equ84}
\mu^{(1)}_{\alpha}(N(\alpha)) \equiv 0 \mod \alpha,
\end{equation}
and
\begin{equation}\label{equ85}
\mu^{(2)}_{\alpha}(N(\alpha)) = 0.
\end{equation}
\end{proposition}

\begin{proof}
Let $\alpha=\alpha_{1}+\alpha_{2}i+\alpha_{3}j+\alpha_{4}k$ be a prime Hurwitz integer. If $\alpha_{1},\alpha_{2},\alpha_{3},\alpha_{4} \in \mathbb{Z},$ then, from \hyperref[equ4]{eq. \ref{equ4}},
\begin{equation}\label{equ86}
\begin{array}{lll}
  \mu^{(1)}_{\alpha}(N(\alpha)) & = & N(\alpha) - \alpha \lfloor \frac{\overline \alpha N(\alpha)}{N(\alpha)} \rceil \\
    & = &  N(\alpha) - \alpha \lfloor \overline \alpha \rceil.
\end{array}
\end{equation}
From \hyperref[equ35]{eq. \ref{equ35}},
\begin{equation}\label{equ87}
\begin{array}{lll}
  \mu^{(1)}_{\alpha}(N(\alpha)) & = & N(\alpha) - \alpha \overline \alpha  \\
    & = &  N(\alpha) - N(\alpha) \\
    & = & 0.
\end{array}
\end{equation}
From \hyperref[equ5]{eq. \ref{equ5}},
\begin{equation}\label{equ89}
\begin{array}{lll}
  \mu^{(2)}_{\alpha}(N(\alpha)) & = & N(\alpha) - \alpha \lfloor \lfloor \frac{\overline \alpha N(\alpha)}{N(\alpha)} \rceil \rceil \\
    & = &  N(\alpha) - \alpha \lfloor \lfloor \overline \alpha \rceil \rceil.
\end{array}
\end{equation}
From \hyperref[equ36]{eq. \ref{equ36}},
\begin{equation}\label{equ90}
\begin{array}{lll}
  \mu^{(2)}_{\alpha}(N(\alpha)) & = & N(\alpha) - \alpha \left( \overline \alpha + \frac{1}{2}+\frac{1}{2}i+\frac{1}{2}j+\frac{1}{2}k \right)  \\
    & = &  N(\alpha) - N(\alpha) - \alpha \left( \frac{1}{2}+\frac{1}{2}i+\frac{1}{2}j+\frac{1}{2}k \right) \\
    & = & - \alpha \left( \frac{1}{2}+\frac{1}{2}i+\frac{1}{2}j+\frac{1}{2}k \right) \\
    & \equiv & 0 \mod \alpha
\end{array}
\end{equation}
If $\alpha_{1},\alpha_{2},\alpha_{3},\alpha_{4} \in \mathbb{Z}+\frac{1}{2},$ then, from \hyperref[equ4]{eq. \ref{equ4}},
\begin{equation}\label{equ91}
\begin{array}{lll}
  \mu^{(1)}_{\alpha}(N(\alpha)) & = & N(\alpha) - \alpha \lfloor \frac{\overline \alpha N(\alpha)}{N(\alpha)} \rceil \\
    & = &  N(\alpha) - \alpha \lfloor \overline \alpha \rceil.
\end{array}
\end{equation}
From \hyperref[equ37]{eq. \ref{equ37}},
\begin{equation}\label{equ92}
\begin{array}{lll}
  \mu^{(1)}_{\alpha}(N(\alpha)) & = & N(\alpha) - \alpha \left( \overline \alpha + \frac{1}{2}+\frac{1}{2}i+\frac{1}{2}j+\frac{1}{2}k \right)  \\
    & = &  N(\alpha) - N(\alpha) - \alpha \left( \frac{1}{2}+\frac{1}{2}i+\frac{1}{2}j+\frac{1}{2}k \right) \\
    & = & - \alpha \left( \frac{1}{2}+\frac{1}{2}i+\frac{1}{2}j+\frac{1}{2}k \right) \\
    & \equiv & 0 \mod \alpha.
\end{array}
\end{equation}
From \hyperref[equ5]{eq. \ref{equ5}},
\begin{equation}\label{equ93}
\begin{array}{lll}
  \mu^{(2)}_{\alpha}(N(\alpha)) & = & N(\alpha) - \alpha \lfloor \lfloor \frac{\overline \alpha N(\alpha)}{N(\alpha)} \rceil \rceil \\
    & = &  N(\alpha) - \alpha \lfloor \lfloor \overline \alpha \rceil \rceil.
\end{array}
\end{equation}
From \hyperref[equ38]{eq. \ref{equ38}},
\begin{equation}\label{equ94}
\begin{array}{lll}
  \mu^{(2)}_{\alpha}(N(\alpha)) & = & N(\alpha) - \alpha \overline \alpha  \\
    & = &  N(\alpha) - N(\alpha) \\
    & = & 0.
\end{array}
\end{equation}
This completes the proof.
\end{proof}

\begin{proposition}\label{prop7}
Let $\alpha=\alpha_{1}+\alpha_{2}i+\alpha_{3}j+\alpha_{4}k$ be a prime Hurwitz integer. Then,
\begin{equation}\label{}
\mu_{\alpha}(z)+\mu_{\alpha}(N(\alpha)-z)\equiv 0 \mod \alpha
\end{equation}
where $z\in \mathbb{Z}_{N(\alpha)}.$
\end{proposition}

\begin{proof}
Let $\alpha$ be a prime Hurwitz integer. Firstly, we suppose that $\mu_{\alpha}(z)=\mu^{(1)}_{\alpha}(z)$ and $\mu_{\alpha} (N(\alpha)-z) = \mu^{(1)}_{\alpha}(N(\alpha)-z)$ where $z\in\mathbb{Z}_{N(\alpha)}.$ Then,
\begin{equation}\label{equ95}
\begin{array}{lll}
\mu_{\alpha}(z)+\mu_{\alpha}(N(\alpha)-z) & = & \mu^{(1)}_{\alpha}(z) + \mu^{(1)}_{\alpha}(N(\alpha)-z).
\end{array}
\end{equation}
From \hyperref[equ4]{eq. \ref{equ4}},
\begin{equation}\label{equ96}
\begin{array}{lll}
  \mu_{\alpha}(z)+\mu_{\alpha}(N(\alpha)-z) & = & z-\alpha\lfloor\frac{\overline \alpha z}{N(\alpha)}\rceil + N(\alpha)-z - \alpha\lfloor \frac{\overline \alpha (N(\alpha)-z)}{N(\alpha)} \rceil \\
  & = & N(\alpha) - \alpha\lfloor\frac{\overline \alpha z}{N(\alpha)}\rceil - \alpha\lfloor \frac{\overline \alpha N(\alpha)}{N(\alpha)} \rceil - \alpha \lfloor \frac{\overline \alpha (-z)}{N(\alpha)} \rceil \\
  & = & N(\alpha) - \alpha\lfloor\frac{\overline \alpha z}{N(\alpha)}\rceil - \alpha\lfloor \overline \alpha \rceil + \alpha \lfloor \frac{\overline \alpha z}{N(\alpha)} \rceil \\
  & = & N(\alpha) - \alpha\lfloor \overline \alpha \rceil.
\end{array}
\end{equation}
If $\alpha_{1},\alpha_{2},\alpha_{3},\alpha_{4} \in \mathbb{Z},$ then, from \hyperref[equ86]{eq. \ref{equ86}} and \hyperref[equ87]{eq. \ref{equ87}},
\begin{equation}\label{equ97}
\begin{array}{lll}
  \mu_{\alpha}(z)+\mu_{\alpha}(N(\alpha)-z) & = & 0 \\
  & \equiv & 0 \mod \alpha.
\end{array}
\end{equation}
If $\alpha_{1},\alpha_{2},\alpha_{3},\alpha_{4} \in \mathbb{Z}+\frac{1}{2},$ then, from \hyperref[equ91]{eq. \ref{equ91}} and \hyperref[equ92]{eq. \ref{equ92}},
\begin{equation}\label{equ98}
\begin{array}{lll}
\mu_{\alpha}(z)+\mu_{\alpha}(N(\alpha)-z) & \equiv & 0 \mod \alpha.
\end{array}
\end{equation}
Secondly, we suppose that $\mu_{\alpha}(z)=\mu^{(1)}_{\alpha}(z)$ and $\mu_{\alpha} (N(\alpha)-z) = \mu^{(2)}_{\alpha}(N(\alpha)-z)$ where $z\in\mathbb{Z}_{N(\alpha)}.$ Then,
\begin{equation}\label{equ99}
\begin{array}{lll}
  \mu_{\alpha}(z)+\mu_{\alpha}(N(\alpha)-z) & = & \mu^{(1)}_{\alpha}(z) + \mu^{(2)}_{\alpha}(N(\alpha)-z).
\end{array}
\end{equation}
From \hyperref[equ4]{eq. \ref{equ4}} and \hyperref[equ5]{eq. \ref{equ5}},
\begin{equation}\label{equ100}
\begin{array}{lll}
  \mu_{\alpha}(z)+\mu_{\alpha}(N(\alpha)-z) & = & z-\alpha\lfloor\frac{\overline \alpha z}{N(\alpha)}\rceil + N(\alpha)-z - \alpha \lfloor \lfloor \frac{\overline \alpha (N(\alpha)-z)}{N(\alpha)} \rceil \rceil \\
  & = & N(\alpha) - \alpha \lfloor \frac{\overline \alpha z}{N(\alpha)} \rceil - \alpha \lfloor\lfloor \frac{\overline \alpha N(\alpha)}{N(\alpha)} \rceil \rceil - \alpha \lfloor \lfloor \frac{\overline \alpha (-z)}{N(\alpha)} \rceil \rceil \\
  & = & N(\alpha) - \alpha \lfloor \frac{\overline \alpha z}{N(\alpha)} \rceil - \alpha \lfloor \lfloor \overline \alpha \rceil \rceil + \alpha \lfloor \lfloor \frac{\overline \alpha z}{N(\alpha)} \rceil \rceil.
\end{array}
\end{equation}
If $\alpha_{1},\alpha_{2},\alpha_{3},\alpha_{4} \in \mathbb{Z},$ then, from \hyperref[equ89]{eq. \ref{equ89}} and \hyperref[equ90]{eq. \ref{equ90}},
\begin{equation}\label{equ101}
\begin{array}{lll}
  \mu_{\alpha}(z)+\mu_{\alpha}(N(\alpha)-z) & = & - \alpha \lfloor \frac{\overline \alpha z}{N(\alpha)} \rceil + \alpha \lfloor \lfloor \frac{\overline \alpha z}{N(\alpha)} \rceil \rceil.
\end{array}
\end{equation}
There exist $\lambda_{1},\lambda_{2}\in \mathcal{H}$ such that $\lfloor \frac{\overline \alpha z}{N(\alpha)} \rceil = \lambda_{1} $ and $\alpha \lfloor \lfloor \frac{\overline \alpha z}{N(\alpha)} \rceil \rceil=\lambda_{2}.$ Hereby,
\begin{equation}\label{equ102}
\begin{array}{lll}
  \mu_{\alpha}(z)+\mu_{\alpha}(N(\alpha)-z) & = & - \alpha \lambda_{1} + \alpha \lambda_{2} \\
  & = & \alpha (-\lambda_{1} + \lambda_{2}) \\
  & \equiv & 0 \mod \alpha.
\end{array}
\end{equation}
If $\alpha_{1},\alpha_{2},\alpha_{3},\alpha_{4} \in \mathbb{Z}+\frac{1}{2},$ then, from \hyperref[equ93]{eq. \ref{equ93}} and \hyperref[equ94]{eq. \ref{equ94}},
\begin{equation}\label{equ103}
\begin{array}{lll}
  \mu_{\alpha}(z)+\mu_{\alpha}(N(\alpha)-z) & = & - \alpha \lfloor \frac{\overline \alpha z}{N(\alpha)} \rceil + \alpha \lfloor \lfloor \frac{\overline \alpha z}{N(\alpha)} \rceil \rceil.
\end{array}
\end{equation}
There exist $\lambda_{1},\lambda_{2}\in \mathcal{H}$ such that $\lfloor \frac{\overline \alpha z}{N(\alpha)} \rceil = \lambda_{1} $ and $\alpha \lfloor \lfloor \frac{\overline \alpha z}{N(\alpha)} \rceil \rceil=\lambda_{2}.$ Hereby,
\begin{equation}\label{equ104}
\begin{array}{lll}
  \mu_{\alpha}(z)+\mu_{\alpha}(N(\alpha)-z) & = & - \alpha \lambda_{1} + \alpha \lambda_{2} \\
  & = & \alpha (-\lambda_{1} + \lambda_{2}) \\
  & \equiv & 0 \mod \alpha.
\end{array}
\end{equation}
Thirdly, we suppose that $\mu_{\alpha}(z)=\mu^{(2)}_{\alpha}(z)$ and $\mu_{\alpha} (N(\alpha)-z) = \mu^{(1)}_{\alpha}(N(\alpha)-z)$ where $z\in\mathbb{Z}_{N(\alpha)}.$ Then,
\begin{equation}\label{equ105}
\begin{array}{lll}
  \mu_{\alpha}(z)+\mu_{\alpha}(N(\alpha)-z) & = & \mu^{(2)}_{\alpha}(z) + \mu^{(1)}_{\alpha}(N(\alpha)-z).
\end{array}
\end{equation}
From \hyperref[equ5]{eq. \ref{equ5}} and \hyperref[equ4]{eq. \ref{equ4}},
\begin{equation}\label{equ106}
\begin{array}{lll}
  \mu_{\alpha}(z)+\mu_{\alpha}(N(\alpha)-z) & = & z-\alpha\lfloor\lfloor\frac{\overline \alpha z}{N(\alpha)}\rceil\rceil + N(\alpha)-z - \alpha \lfloor \frac{\overline \alpha (N(\alpha)-z)}{N(\alpha)} \rceil \\
  & = & N(\alpha) - \alpha \lfloor \lfloor \frac{\overline \alpha z}{N(\alpha)} \rceil \rceil - \alpha \lfloor \frac{\overline \alpha N(\alpha)}{N(\alpha)} \rceil - \alpha \lfloor \frac{\overline \alpha (-z)}{N(\alpha)} \rceil \\
  & = & N(\alpha) - \alpha \lfloor \lfloor \frac{\overline \alpha z}{N(\alpha)} \rceil \rceil - \alpha \lfloor \overline \alpha \rceil + \alpha \lfloor \frac{\overline \alpha z}{N(\alpha)} \rceil.
\end{array}
\end{equation}
If $\alpha_{1},\alpha_{2},\alpha_{3},\alpha_{4} \in \mathbb{Z},$ then, from \hyperref[equ86]{eq. \ref{equ86}} and \hyperref[equ87]{eq. \ref{equ87}},
\begin{equation}\label{equ107}
\begin{array}{lll}
  \mu_{\alpha}(z)+\mu_{\alpha}(N(\alpha)-z) & = & - \alpha \lfloor \lfloor \frac{\overline \alpha z}{N(\alpha)} \rceil \rceil + \alpha \lfloor \frac{\overline \alpha z}{N(\alpha)} \rceil.
\end{array}
\end{equation}
There exist $\lambda_{1},\lambda_{2}\in \mathcal{H}$ such that $\lfloor\lfloor \frac{\overline \alpha z}{N(\alpha)} \rceil \rceil = \lambda_{1} $ and $\alpha \lfloor \frac{\overline \alpha z}{N(\alpha)} \rceil=\lambda_{2}.$ Hereby,
\begin{equation}\label{equ108}
\begin{array}{lll}
  \mu_{\alpha}(z)+\mu_{\alpha}(N(\alpha)-z) & = & - \alpha \lambda_{1} + \alpha \lambda_{2} \\
  & = & \alpha (-\lambda_{1} + \lambda_{2}) \\
  & \equiv & 0 \mod \alpha.
\end{array}
\end{equation}
If $\alpha_{1},\alpha_{2},\alpha_{3},\alpha_{4} \in \mathbb{Z}+\frac{1}{2},$ then, from \hyperref[equ91]{eq. \ref{equ91}} and \hyperref[equ92]{eq. \ref{equ92}},
\begin{equation}\label{equ109}
\begin{array}{lll}
  \mu_{\alpha}(z)+\mu_{\alpha}(N(\alpha)-z) & = & - \alpha \lfloor \lfloor \frac{\overline \alpha z}{N(\alpha)} \rceil \rceil + \alpha \lfloor \frac{\overline \alpha z}{N(\alpha)} \rceil.
\end{array}
\end{equation}
There exist $\lambda_{1},\lambda_{2}\in \mathcal{H}$ such that $\lfloor \lfloor \frac{\overline \alpha z}{N(\alpha)} \rceil \rceil = \lambda_{1} $ and $\alpha \lfloor \frac{\overline \alpha z}{N(\alpha)} \rceil=\lambda_{2}.$ Hereby,
\begin{equation}\label{equ110}
\begin{array}{lll}
  \mu_{\alpha}(z)+\mu_{\alpha}(N(\alpha)-z) & = & - \alpha \lambda_{1} + \alpha \lambda_{2} \\
  & = & \alpha (-\lambda_{1} + \lambda_{2}) \\
  & \equiv & 0 \mod \alpha.
\end{array}
\end{equation}
Lastly, we suppose that $\mu_{\alpha}(z)=\mu^{(2)}_{\alpha}(z)$ and $\mu_{\alpha} (N(\alpha)-z) = \mu^{(2)}_{\alpha}(N(\alpha)-z)$ where $z\in\mathbb{Z}_{N(\alpha)}.$ Then,
\begin{equation}\label{equ11}
\begin{array}{lll}
\mu_{\alpha}(z)+\mu_{\alpha}(N(\alpha)-z) & = & \mu^{(2)}_{\alpha}(z) + \mu^{(12)}_{\alpha}(N(\alpha)-z).
\end{array}
\end{equation}
From \hyperref[equ4]{eq. \ref{equ4}},
\begin{equation}\label{equ112}
\begin{array}{lll}
  \mu_{\alpha}(z)+\mu_{\alpha}(N(\alpha)-z) & = & z-\alpha\lfloor\lfloor\frac{\overline \alpha z}{N(\alpha)}\rceil\rceil + N(\alpha)-z - \alpha\lfloor\lfloor \frac{\overline \alpha (N(\alpha)-z)}{N(\alpha)} \rceil \rceil \\
  & = & N(\alpha) - \alpha\lfloor\lfloor\frac{\overline \alpha z}{N(\alpha)}\rceil\rceil - \alpha\lfloor\lfloor \frac{\overline \alpha N(\alpha)}{N(\alpha)} \rceil\rceil - \alpha \lfloor\lfloor \frac{\overline \alpha (-z)}{N(\alpha)} \rceil\rceil \\
  & = & N(\alpha) - \alpha\lfloor\lfloor\frac{\overline \alpha z}{N(\alpha)}\rceil\rceil - \alpha\lfloor\lfloor \overline \alpha \rceil \rceil + \alpha \lfloor \lfloor \frac{\overline \alpha z}{N(\alpha)} \rceil \rceil \\
  & = & N(\alpha) - \alpha \lfloor \lfloor \overline \alpha \rceil \rceil.
\end{array}
\end{equation}
If $\alpha_{1},\alpha_{2},\alpha_{3},\alpha_{4} \in \mathbb{Z},$ then, from \hyperref[equ89]{eq. \ref{equ89}} and \hyperref[equ90]{eq. \ref{equ90}},
\begin{equation}\label{equ113}
\begin{array}{lll}
  \mu_{\alpha}(z)+\mu_{\alpha}(N(\alpha)-z) & = & 0 \\
  & \equiv & 0 \mod \alpha.
\end{array}
\end{equation}
If $\alpha_{1},\alpha_{2},\alpha_{3},\alpha_{4} \in \mathbb{Z}+\frac{1}{2},$ then, from \hyperref[equ93]{eq. \ref{equ93}} and \hyperref[equ94]{eq. \ref{equ94}},
\begin{equation}\label{equ114}
\begin{array}{lll}
\mu_{\alpha}(z)+\mu_{\alpha}(N(\alpha)-z) & \equiv & 0 \mod \alpha.
\end{array}
\end{equation}
\end{proof}
Consequently, $\mu_{\alpha}(z)+\mu_{\alpha}(N(\alpha)-z)\equiv 0 \mod \alpha.$ This completes the proof.

\section{Examples}\label{sec5}

\begin{example}
Let $\alpha=\frac{5}{2}+\frac{3}{2}i+\frac{3}{2}j+\frac{3}{2}k.$ $\alpha$ is a prime Hurwitz integer because of $N(\alpha)=13.$ From \hyperref[equ4]{eq. \ref{equ4}},
\begin{equation}\label{equ115}
\mathcal{H}^{(1)}_{\alpha} = \left\{ \begin{gathered} \mu^{(1)}_{\alpha}(0)=0,\mu^{(1)}_{\alpha}(1)=1,\mu^{(1)}_{\alpha}(2)=2,\mu^{(1)}_{\alpha}(3)=\frac{1}{2}-\frac{3}{2}i-\frac{3}{2}j-\frac{3}{2}k,\hfill \\
\mu^{(1)}_{\alpha}(4)=\frac{3}{2}-\frac{3}{2}i-\frac{3}{2}j-\frac{3}{2}k,\mu^{(1)}_{\alpha}(5)=-2+i+j+k,\mu^{(1)}_{\alpha}(6)=-1+i+j+k,\hfill \\
\mu^{(1)}_{\alpha}(7)=i+j+k,\mu^{(1)}_{\alpha}(8)=-\frac{3}{2}-\frac{1}{2}i-\frac{1}{2}j-\frac{1}{2}k,\mu^{(1)}_{\alpha}(9)=-\frac{1}{2}-\frac{1}{2}i-\frac{1}{2}j-\frac{1}{2}k,\hfill \\
\mu^{(1)}_{\alpha}(10)=\frac{1}{2}-\frac{1}{2}i-\frac{1}{2}j-\frac{1}{2}k,\mu^{(1)}_{\alpha}(11)=\frac{3}{2}-\frac{1}{2}i-\frac{1}{2}j-\frac{1}{2}k, \hfill \\
\mu^{(1)}_{\alpha}(12)=\frac{5}{2}-\frac{1}{2}i-\frac{1}{2}j-\frac{1}{2}k \hfill \end{gathered} \right\}.
\end{equation}
From \hyperref[equ5]{eq. \ref{equ5}},
\begin{equation}\label{equ116}
\mathcal{H}^{(2)}_{\alpha} = \left\{ \begin{gathered} \mu^{(2)}_{\alpha}(0)=1-2i-2j-2k,\mu^{(2)}_{\alpha}(1)=-\frac{5}{2}+\frac{1}{2}i+\frac{1}{2}j+\frac{1}{2}k, \hfill \\
\mu^{(2)}_{\alpha}(2)=-\frac{3}{2}+\frac{1}{2}i+\frac{1}{2}j+\frac{1}{2}k,\mu^{(2)}_{\alpha}(3)=-\frac{1}{2}+\frac{1}{2}i+\frac{1}{2}j+\frac{1}{2}k, \hfill \\
\mu^{(2)}_{\alpha}(4)=\frac{1}{2}+\frac{1}{2}i+\frac{1}{2}j+\frac{1}{2}k,\mu^{(2)}_{\alpha}(5)=\frac{3}{2}+\frac{1}{2}i+\frac{1}{2}j+\frac{1}{2}k, \hfill \\
\mu^{(2)}_{\alpha}(6)=-i-j-k, \mu^{(2)}_{\alpha}(7)=1-i-j-k, \mu^{(2)}_{\alpha}(8)=2-i-j-k,\hfill \\
\mu^{(2)}_{\alpha}(9)=-\frac{3}{2}+\frac{3}{2}i+\frac{3}{2}j+\frac{3}{2}k,\mu^{(2)}_{\alpha}(10)=-\frac{1}{2}+\frac{3}{2}i+\frac{3}{2}j+\frac{3}{2}k,\hfill \\
\mu^{(2)}_{\alpha}(11)=-2, \mu^{(2)}_{\alpha}(12)=-1 \hfill \end{gathered} \right\}.
\end{equation}
With respect to \hyperref[equ6]{eq. \ref{equ6}}
\begin{equation}\label{equ117}
\mathcal{H}_{\alpha} = \left\{ \begin{gathered} \mu_{\alpha}(0)=\mu^{(1)}_{\alpha}(0)=0,\mu_{\alpha}(1)=\mu^{(1)}_{\alpha}(1)=1,\mu_{\alpha}(2)=\mu^{(2)}_{\alpha}(2)=-\frac{3}{2}+\frac{1}{2}i+\frac{1}{2}j+\frac{1}{2}k,\hfill \\
\mu_{\alpha}(3)=\mu^{(2)}_{\alpha}(3)=-\frac{1}{2}+\frac{1}{2}i+\frac{1}{2}j+\frac{1}{2}k,\mu_{\alpha}(4)=\mu^{(2)}_{\alpha}(4)=\frac{1}{2}+\frac{1}{2}i+\frac{1}{2}j+\frac{1}{2}k,\hfill \\
\mu_{\alpha}(5)=\mu^{(2)}_{\alpha}(5)=\frac{3}{2}+\frac{1}{2}i+\frac{1}{2}j+\frac{1}{2}k,\mu_{\alpha}(6)=\mu^{(2)}_{\alpha}(6)=-i-j-k, \hfill \\
\mu_{\alpha}(7)=\mu^{(1)}_{\alpha}(7)=i+j+k,\mu_{\alpha}(8)=\mu^{(1)}_{\alpha}(8)=-\frac{3}{2}-\frac{1}{2}i-\frac{1}{2}j-\frac{1}{2}k,\hfill \\
\mu_{\alpha}(9)=\mu^{(1)}_{\alpha}(9)=-\frac{1}{2}-\frac{1}{2}i-\frac{1}{2}j-\frac{1}{2}k,\mu_{\alpha}(10)=\mu^{(1)}_{\alpha}(10)=\frac{1}{2}-\frac{1}{2}i-\frac{1}{2}j-\frac{1}{2}k, \hfill \\
\mu_{\alpha}(11)=\mu^{(1)}_{\alpha}(11)=\frac{3}{2}-\frac{1}{2}i-\frac{1}{2}j-\frac{1}{2}k,\mu_{\alpha}(12)=\mu^{(2)}_{\alpha}(12)=-1 \hfill \end{gathered} \right\}.
\end{equation}
From \hyperref[equ127]{eq. \ref{equ127}}, the average energy of $\mathcal{H}_{\frac{5}{2} + \frac{3}{2}i + \frac{3}{2}j + \frac{3}{2}k}$ is
\begin{equation}
\mathcal{E}_{\frac{5}{2} + \frac{3}{2}i + \frac{3}{2}j + \frac{3}{2}k} = \frac{24}{13} = 1.8462.
\end{equation}
\end{example}

\begin{example}
Let $\alpha=3+2i.$ $\alpha$ is a prime Hurwitz integer because of $N(\alpha)=13.$ From \hyperref[equ4]{eq. \ref{equ4}},
\begin{equation}\label{equ118}
\mathcal{H}^{(1)}_{\alpha} = \left\{ \begin{gathered} \mu^{(1)}_{\alpha}(0)=0,\mu^{(1)}_{\alpha}(1)=1,\mu^{(1)}_{\alpha}(2)=2,\mu^{(1)}_{\alpha}(3)=-2i,
\mu^{(1)}_{\alpha}(4)=-1+i,\mu^{(1)}_{\alpha}(5)=i, \hfill \\
\mu^{(1)}_{\alpha}(6)=1+i,\mu^{(1)}_{\alpha}(7)=-1-i,\mu^{(1)}_{\alpha}(8)=-i,\mu^{(1)}_{\alpha}(9)=1-i,\mu^{(1)}_{\alpha}(10)=2i, \hfill \\
\mu^{(1)}_{\alpha}(11)=-2, \mu^{(1)}_{\alpha}(12)=-1 \hfill \end{gathered} \right\}.
\end{equation}
From \hyperref[equ5]{eq. \ref{equ5}},
\begin{equation}\label{equ119}
\mathcal{H}^{(2)}_{\alpha} = \left\{ \begin{gathered} \mu^{(2)}_{\alpha}(0)=-\frac{1}{2}-\frac{5}{2}i-\frac{1}{2}j-\frac{5}{2}k,\mu^{(2)}_{\alpha}(1)=-\frac{3}{2}+\frac{1}{2}i-\frac{1}{2}j-\frac{5}{2}k, \hfill \\
\mu^{(2)}_{\alpha}(2)=-\frac{1}{2}+\frac{1}{2}i-\frac{1}{2}j-\frac{5}{2}k,\mu^{(2)}_{\alpha}(3)=\frac{1}{2}+\frac{1}{2}i-\frac{1}{2}j-\frac{5}{2}k, \hfill \\
\mu^{(2)}_{\alpha}(4)=\frac{3}{2}+\frac{1}{2}i-\frac{1}{2}j-\frac{5}{2}k,\mu^{(2)}_{\alpha}(5)=-\frac{1}{2}-\frac{3}{2}i-\frac{1}{2}j-\frac{5}{2}k, \hfill \\
\mu^{(2)}_{\alpha}(6)=\frac{1}{2}-\frac{3}{2}i-\frac{1}{2}j-\frac{5}{2}k,\mu^{(2)}_{\alpha}(7)=-\frac{1}{2}+\frac{3}{2}i-\frac{1}{2}j-\frac{5}{2}k, \hfill \\
\mu^{(2)}_{\alpha}(8)=\frac{1}{2}+\frac{3}{2}i-\frac{1}{2}j-\frac{5}{2}k,\mu^{(2)}_{\alpha}(9)=-\frac{3}{2}-\frac{1}{2}i-\frac{1}{2}j-\frac{5}{2}k, \hfill \\
\mu^{(2)}_{\alpha}(10)=-\frac{1}{2}-\frac{1}{2}i-\frac{1}{2}j-\frac{5}{2}k,\mu^{(2)}_{\alpha}(11)=\frac{1}{2}-\frac{1}{2}i-\frac{1}{2}j-\frac{5}{2}k, \hfill \\
\mu^{(2)}_{\alpha}(12)=\frac{3}{2}-\frac{1}{2}i-\frac{1}{2}j-\frac{5}{2}k \hfill \end{gathered} \right\}.
\end{equation}
With respect to \hyperref[equ6]{eq. \ref{equ6}}, and \hyperref[cor1]{corollary \ref{cor1}}
\begin{equation}\label{equ120}
\mathcal{H}_{\alpha} = \mathcal{H}^{(1)}_{\alpha}
\end{equation}
From \hyperref[equ127]{eq. \ref{equ127}}, the average energy of $\mathcal{H}_{3+2i}$ is
\begin{equation}
\mathcal{E}_{3+2i} = \frac{28}{13} = 2.1539.
\end{equation}
\end{example}

The average energy for the transmitted signal, considering the sets of residual class with the same cardinality,
the average energy of $\mathcal{H}_{\frac{5}{2} + \frac{3}{2}i + \frac{3}{2}j + \frac{3}{2}k}$ is smaller than the average energy of $\mathcal{H}_{3+2i}.$

\begin{example}
Let $\alpha=3+i+j.$ $\alpha$ is a prime Hurwitz integer because of $N(\alpha)=11.$ From \hyperref[equ4]{eq. \ref{equ4}},
\begin{equation}\label{equ121}
\mathcal{H}^{(1)}_{\alpha} = \left\{ \begin{gathered} \mu^{(1)}_{\alpha}(0)=0,\mu^{(1)}_{\alpha}(1)=1,\mu^{(1)}_{\alpha}(2)=-1-i-j,\mu^{(1)}_{\alpha}(3)=-i-j,
\mu^{(1)}_{\alpha}(4)=1-i-j, \hfill \\
\mu^{(1)}_{\alpha}(5)=2-i-j,\mu^{(1)}_{\alpha}(6)=-2+i+j,\mu^{(1)}_{\alpha}(7)=-1+i+j,\mu^{(1)}_{\alpha}(8)=i+j,\hfill \\
\mu^{(1)}_{\alpha}(9)=1+i+j,\mu^{(1)}_{\alpha}(10)=-1, \hfill \end{gathered} \right\}.
\end{equation}
From \hyperref[equ5]{eq. \ref{equ5}},
\begin{equation}\label{equ122}
\mathcal{H}^{(2)}_{\alpha} = \left\{ \begin{gathered} \mu^{(2)}_{\alpha}(0)=-\frac{1}{2}-\frac{5}{2}i-\frac{3}{2}j-\frac{3}{2}k,\mu^{(2)}_{\alpha}(1)=-\frac{3}{2}+\frac{1}{2}i+\frac{3}{2}j-\frac{3}{2}k, \hfill \\
\mu^{(2)}_{\alpha}(2)=-\frac{1}{2}+\frac{1}{2}i+\frac{3}{2}j-\frac{3}{2}k,\mu^{(2)}_{\alpha}(3)=\frac{1}{2}+\frac{1}{2}i+\frac{3}{2}j-\frac{3}{2}k, \hfill \\
\mu^{(2)}_{\alpha}(4)=-\frac{3}{2}-\frac{1}{2}i+\frac{1}{2}j-\frac{3}{2}k,\mu^{(2)}_{\alpha}(5)=-\frac{1}{2}-\frac{1}{2}i+\frac{1}{2}j-\frac{3}{2}k, \hfill \\
\mu^{(2)}_{\alpha}(6)=\frac{1}{2}-\frac{1}{2}i+\frac{1}{2}j-\frac{3}{2}k,\mu^{(2)}_{\alpha}(7)=\frac{3}{2}-\frac{1}{2}i+\frac{1}{2}j+\frac{3}{2}k, \hfill \\
\mu^{(2)}_{\alpha}(8)=-\frac{1}{2}-\frac{3}{2}i-\frac{1}{2}j-\frac{3}{2}k,\mu^{(2)}_{\alpha}(9)=\frac{1}{2}-\frac{3}{2}i-\frac{1}{2}j-\frac{3}{2}k, \hfill \\
\mu^{(2)}_{\alpha}(10)=\frac{3}{2}-\frac{3}{2}i-\frac{1}{2}j-\frac{3}{2}k, \hfill \end{gathered} \right\}.
\end{equation}
With respect to \hyperref[equ6]{eq. \ref{equ6}}
\begin{equation}\label{equ123}
\mathcal{H}_{\alpha} = \left\{ \begin{gathered} \mu_{\alpha}(0)=\mu^{(1)}_{\alpha}(0)=0,\mu_{\alpha}(1)=\mu^{(1)}_{\alpha}(1)=1,\mu_{\alpha}(2)=\mu^{(1)}_{\alpha}(2)=-1-i-j,\hfill \\
\mu_{\alpha}(3)=\mu^{(1)}_{\alpha}(3)=-i-j,\mu_{\alpha}(4)=\mu^{(1)}_{\alpha}(4)=1-i-j,\hfill \\
\mu_{\alpha}(5)=\mu^{(2)}_{\alpha}(5)=-\frac{1}{2}-\frac{1}{2}i+\frac{1}{2}j-\frac{3}{2}k,\mu_{\alpha}(6)=\mu^{(2)}_{\alpha}(6)=\frac{1}{2}-\frac{1}{2}i+\frac{1}{2}j-\frac{3}{2}k,\hfill \\
\mu_{\alpha}(7)=\mu^{(1)}_{\alpha}(7)=-1+i+j,\mu_{\alpha}(8)=\mu^{(1)}_{\alpha}(8)=i+j,\hfill \\
\mu_{\alpha}(9)=\mu^{(1)}_{\alpha}(9)=1+i+j,\mu_{\alpha}(10)=\mu^{(1)}_{\alpha}(10)=-1, \hfill \end{gathered} \right\}.
\end{equation}
From \hyperref[equ127]{eq. \ref{equ127}}, the average energy of $\mathcal{H}_{3+i+j}$ is
\begin{equation}
\mathcal{E}_{3+i+j} = \frac{24}{11} = 2.1818.
\end{equation}
\end{example}

\footnotesize \begin{table} \label{table1}
\begin{center}
\caption{For $N \leq 50,$ the Average Energies of Prime Hurwitz Integers such that $N=6k+1$ where $k \in \mathbb{Z}^{+}.$ }
\renewcommand{\arraystretch}{1.2} 
\begin{tabular}{|c|c|c|c|}
  \hline
  N & Hurwitz Integers ($\mathbb{Z}$) & Hurwitz Integers ($\mathbb{Z}+\frac{1}{2}$) & Average Energy \\ \hline
  $7$  & $2+i+j+k$ & $\frac{3}{2}+\frac{3}{2}i+\frac{3}{2}j+\frac{1}{2}k$ & $0.8571$ \\ \cline{3-3}
       &           & $\frac{5}{2}+\frac{1}{2}i+\frac{1}{2}j+\frac{1}{2}k$ &          \\ \hline
  $13$ & $2+2i+2j+k$ & $\frac{5}{2}+\frac{3}{2}i+\frac{3}{2}j+\frac{3}{2}k$ & $1.8462$ \\ \cline{3-3}
       &             & $\frac{7}{2}+\frac{1}{2}i+\frac{1}{2}j+\frac{1}{2}k$ &          \\ \hline
  $19$ & $4+i+j+k$   & $\frac{5}{2}+\frac{5}{2}i+\frac{5}{2}j+\frac{1}{2}k$ & $2.5263$ \\ \cline{3-3}
       &             & $\frac{7}{2}+\frac{3}{2}i+\frac{3}{2}j+\frac{3}{2}k$ &          \\ \hline
  $31$ & $3+3i+3j+2k$ & $\frac{7}{2}+\frac{5}{2}i+\frac{5}{2}j+\frac{5}{2}k$ & $4.2581$ \\ \cline{3-3}
       &              & $\frac{11}{2}+\frac{1}{2}i+\frac{1}{2}j+\frac{1}{2}k$ &         \\ \hline
  $37$ & $5+2i+2j+2k$ & $\frac{7}{2}+\frac{7}{2}i+\frac{7}{2}j+\frac{1}{2}k$ & $5.0270$ \\ \cline{3-3}
       &              & $\frac{11}{2}+\frac{3}{2}i+\frac{3}{2}j+\frac{3}{2}k$ &         \\ \hline
  $43$ & $4+3i+3j+3k$ & $\frac{7}{2}+\frac{7}{2}i+\frac{7}{2}j+\frac{5}{2}k$ & $6.0000$ \\ \cline{3-3}
       &              & $\frac{13}{2}+\frac{1}{2}i+\frac{1}{2}j+\frac{1}{2}k$ &         \\ \hline
\end{tabular}
\end{center}
\end{table} \normalsize

\hyperref[table1]{Table I} is presented the average energies of prime Hurwitz integers such that $N=6k+1 \leq 50$ where $k \in \mathbb{Z}^{+}.$

\section{The Code Rate for Codes over Hurwitz Integers} \label{sec6}

Let $\beta$ be a Hurwitz integer. $\beta$ is a unit if and only if $N(\beta)=1.$ There are precisely $24$ units in $\mathcal{H}.$  The set of units in $\mathcal{H}$ that denoted by $\mathcal{U}_{\mathcal{H}}$ is shown by
\begin{equation}\label{}
\mathcal{U}_{\mathcal{H}}= \left \{ \pm 1, \pm i, \pm j, \pm k, \pm 1 \pm i \pm j \pm k, \pm \frac{1}{2} \pm \frac{1}{2}i \pm \frac{1}{2}j \pm \frac{1}{2}k \right\}.
\end{equation}
In this section, we give the code rate of a code $C$ over $\mathcal{H}.$ Note that $\mathcal{H}_{\alpha}$ has $N(\alpha)$ elements. Let $N(\alpha)=p$ be a prime integer such that $p \equiv 1 \mod 24.$ A code $C$ over $\mathcal{H}$ has length $n=\frac{p^{k}-1}{24}.$ Here $k$ is the dimension of a code $C$ over $\mathcal{H}.$ Let $|C| = |\mathcal{H}_{\alpha}| = N(\alpha) = p.$ If $|C|=N(\alpha)=p,$ then $k=\log_{p}|C|=\log_{p}p=1.$ The coding rate of a code $C$ over $\mathcal{H}$ is computed by
\begin{equation}\label{128}
R=\frac{k}{n}=\frac{24k}{p^{k}-1}.
\end{equation}
In this study, we consider $k=1.$
\begin{example}
For $ p \leq 100,$ $p = 73$ and $p = 97$ are prime integers such that $p \equiv 1 \mod 24.$
\begin{itemize}
  \item[i.] For $ N(\alpha) = p = 73,$ $\alpha$ is a prime Hurwitz integer. The length of a code $C$ over $\mathcal{H}$ is $n=\frac{72}{24}=3.$ The rate of a code $C$ over $\mathcal{H}$ is $R=\frac{24}{72}=\frac{1}{3}.$ $C$ is a $(3,1)-$ code over $\mathcal{H}$ because of $n=3$ and $k=1.$
  \item[ii.] For $ N(\alpha) = p = 97,$ $\alpha$ is a prime Hurwitz integer. The length of a code $C$ over $\mathcal{H}$ is $n=\frac{96}{24}=4.$ The rate of a code $C$ over $\mathcal{H}$ is $R=\frac{24}{96}=\frac{1}{4}.$ $C$ is a $(4,1)-$ code over $\mathcal{H}$ because of $n=4$ and $k=1.$
\end{itemize}
\end{example}

\section{Graph Layout Methods} \label{sec7}

Graphs are used to show the relationship between items, in general. Graph drawing enables visualization of these relationships. The usefulness of the visual representation depends upon whether the drawing is aesthetic. While there are no strict criteria for aesthetic drawing, it is generally agreed that such a drawing has minimal edge crossing and even spacing between vertices. Two popular straight-edge drawing algorithms, the spring embedding, and spring-electrical embedding work by minimizing the energy of physical models of the graph. The high-dimensional embedding method, on the other hand, embeds a graph in high-dimensional space and then projects it back to two or three-dimensional space. Random, circular, and spiral embedding do not utilize connectivity information for laying out a graph. In this study, we do not consider random embedding, and circular embedding. We consider spring, the high-dimensional, and spiral embedding methods, in this study. The spring embedding algorithm assigns force between each pair of nodes. When two nodes are too close together, a repelling force comes into effect. When two nodes are too far apart, they are subject to an attractive force. This scenario can be illustrated by linking the vertices with springs, hence the name "spring embedding". In the high-dimensional embedding method, a graph is embedded in high-dimensional space and then projected back to two or three-dimensional space. The high-dimensional embedding method tends to be very fast but its results are often of lower quality than force-directed algorithms. We use the Wolfram Mathematica $10.2$ program, which has algorithms used for layered/hierarchical drawing of directed graphs and, for drawing trees. These algorithms are implemented via four functions: "GraphPlot,"  "GraphPlot3D," "LayeredGraphPlot," and "TreePlot". In this study, we use two of these algorithms for graph drawing, i.e., "GraphPlot,"  and "GraphPlot3D". You can find more details on the \href{https://reference.wolfram.com/language/tutorial/GraphDrawingIntroduction.html#25676449}{wolfram.com}.

\begin{definition}\textbf{Spring Embedding}
The spring embedding is a graph-drawing technique to position vertices of a graph so that they minimize the mechanical energy when each edge corresponds to a spring. The spring embedding is typically used to lay out regular structured graphs. You can find more details on the \href{https://reference.wolfram.com/language/ref/method/SpringEmbedding.html}{wolfram.com}.
\end{definition}

\begin{definition} \textbf{High Dimensional Embedding}
The high-dimensional embedding is a graph-drawing technique to position vertices of a graph in a high-dimensional space, and then project back to two- or three-dimensional space. The high-dimensional embedding is typically used for fast layout of graphs. You can find more details on the \href{https://reference.wolfram.com/language/ref/method/HighDimensionalEmbedding.html}{wolfram.com}.
\end{definition}

\begin{definition} \textbf{Spiral Embedding}
The spiral embedding is a graph-drawing technique to position vertices of a graph on a 3D spiral projected to 2D. The spiral embedding is typically used to lay out path graphs. You can find more details on the \href{https://reference.wolfram.com/language/ref/method/SpiralEmbedding.html}{wolfram.com}.
\end{definition}

\begin{figure}[h!]
  \centering
  \begin{subfigure}[b]{0.2\linewidth}
    \includegraphics[width=\linewidth]{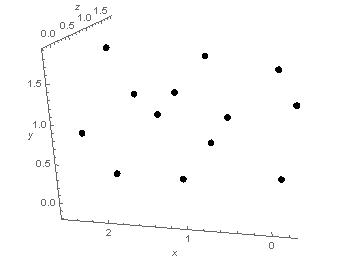}
    \caption{Points on three dimensional coordinate system}
  \end{subfigure}
  \begin{subfigure}[b]{0.2\linewidth}
    \includegraphics[width=\linewidth]{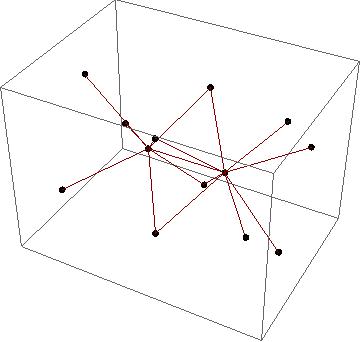}
    \caption{Graph in a three-dimensional box}
  \end{subfigure}
  \begin{subfigure}[b]{0.2\linewidth}
    \includegraphics[width=\linewidth]{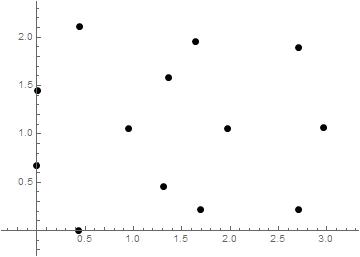}
    \caption{Points on two-dimensional coordinate system}
  \end{subfigure}
  \begin{subfigure}[b]{0.2\linewidth}
    \includegraphics[width=\linewidth]{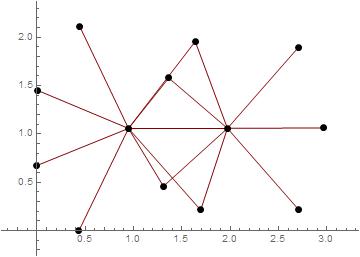}
    \caption{Graph on two-dimensional coordinate system}
  \end{subfigure}
  \caption{The Spring Embedding of $3+2i$ on Two or Three-Dimensional Coordinate System.}
  \label{fig:graph}
\end{figure}

\begin{figure}[h!]
  \centering
  \begin{subfigure}[b]{0.2\linewidth}
    \includegraphics[width=\linewidth]{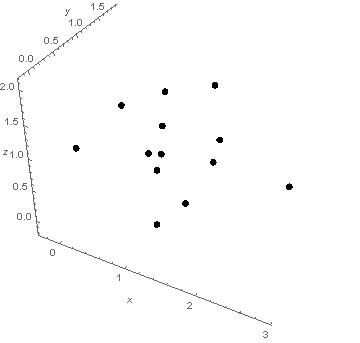}
    \caption{Points on three dimensional coordinate system}
  \end{subfigure}
  \begin{subfigure}[b]{0.2\linewidth}
    \includegraphics[width=\linewidth]{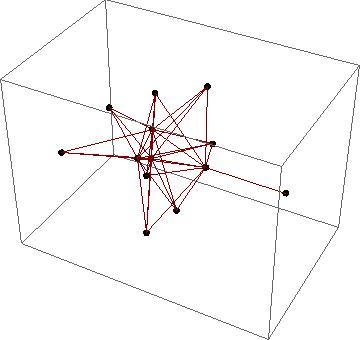}
    \caption{Graph in a three-dimensional box}
  \end{subfigure}
  \begin{subfigure}[b]{0.2\linewidth}
    \includegraphics[width=\linewidth]{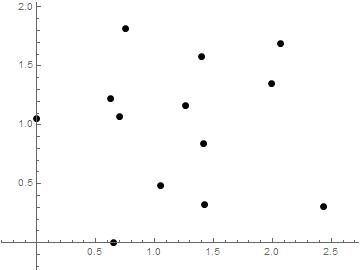}
    \caption{Points on two-dimensional coordinate system}
  \end{subfigure}
  \begin{subfigure}[b]{0.2\linewidth}
    \includegraphics[width=\linewidth]{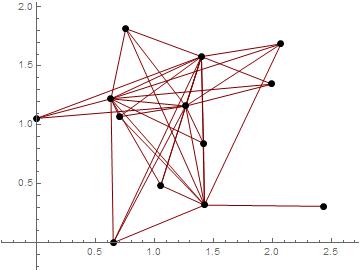}
    \caption{Graph on two-dimensional coordinate system}
  \end{subfigure}
  \caption{The Spring Embedding of $\frac{5}{2}+\frac{3}{2}i+\frac{3}{2}j+\frac{3}{2}k$ on Two or Three-Dimensional Coordinate System.}
  \label{fig:graph}
\end{figure}

\newpage

\begin{figure}[h!]
  \centering
  \begin{subfigure}[b]{0.2\linewidth}
    \includegraphics[width=\linewidth]{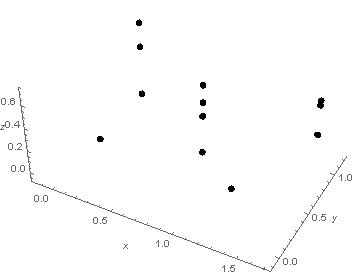}
    \caption{Points on three dimensional coordinate system}
  \end{subfigure}
  \begin{subfigure}[b]{0.2\linewidth}
    \includegraphics[width=\linewidth]{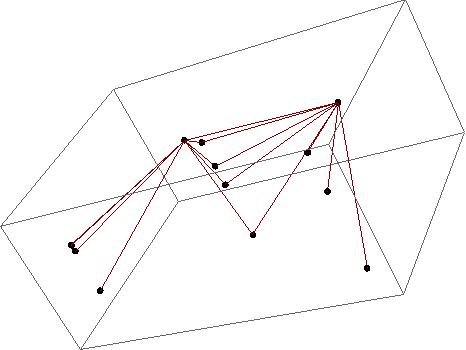}
    \caption{Graph in a three-dimensional box}
  \end{subfigure}
  \begin{subfigure}[b]{0.2\linewidth}
    \includegraphics[width=\linewidth]{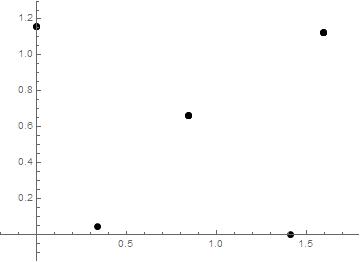}
    \caption{Points on two-dimensional coordinate system}
  \end{subfigure}
  \begin{subfigure}[b]{0.2\linewidth}
    \includegraphics[width=\linewidth]{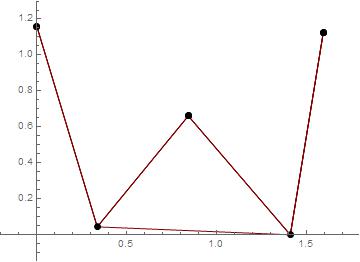}
    \caption{Graph on two-dimensional coordinate system}
  \end{subfigure}
  \caption{The High-Dimensional Embedding of $3+2i$ on Two or Three-Dimensional Coordinate System.}
  \label{fig:graph}
\end{figure}

\begin{figure}[h!]
  \centering
  \begin{subfigure}[b]{0.2\linewidth}
    \includegraphics[width=\linewidth]{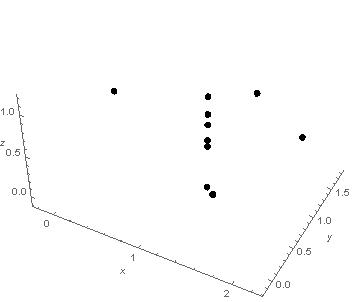}
    \caption{Points on three dimensional coordinate system}
  \end{subfigure}
  \begin{subfigure}[b]{0.2\linewidth}
    \includegraphics[width=\linewidth]{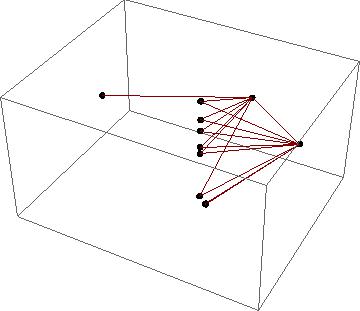}
    \caption{Graph in a three-dimensional box}
  \end{subfigure}
  \begin{subfigure}[b]{0.2\linewidth}
    \includegraphics[width=\linewidth]{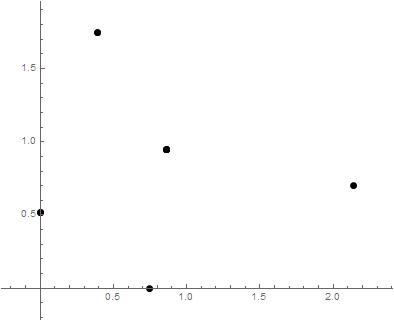}
    \caption{Points on two-dimensional coordinate system}
  \end{subfigure}
  \begin{subfigure}[b]{0.2\linewidth}
    \includegraphics[width=\linewidth]{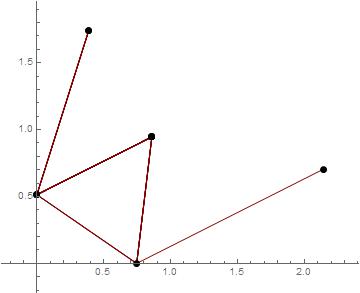}
    \caption{Graph on two-dimensional coordinate system}
  \end{subfigure}
  \caption{The High-Dimensional Embedding of $\frac{5}{2}+\frac{3}{2}i+\frac{3}{2}j+\frac{3}{2}k$ on Two or Three-Dimensional Coordinate System.}
  \label{fig:graph}
\end{figure}

\begin{figure}[h!]
  \centering
  \begin{subfigure}[b]{0.2\linewidth}
    \includegraphics[width=\linewidth]{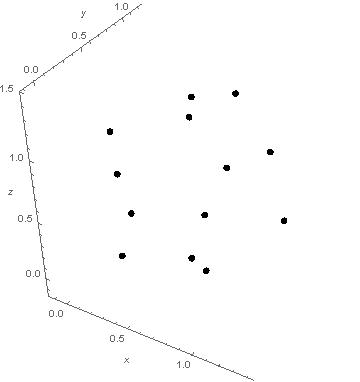}
    \caption{Points on three dimensional coordinate system}
  \end{subfigure}
  \begin{subfigure}[b]{0.2\linewidth}
    \includegraphics[width=\linewidth]{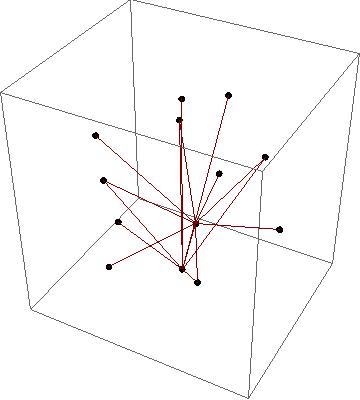}
    \caption{Graph in a three-dimensional box}
  \end{subfigure}
  \begin{subfigure}[b]{0.2\linewidth}
    \includegraphics[width=\linewidth]{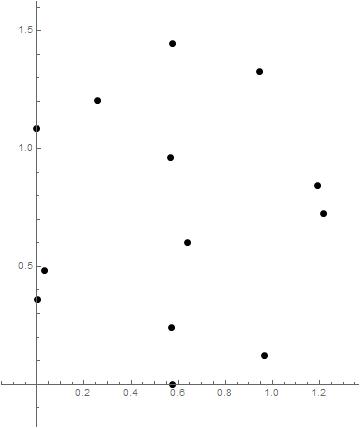}
    \caption{Points on two-dimensional coordinate system}
  \end{subfigure}
  \begin{subfigure}[b]{0.2\linewidth}
    \includegraphics[width=\linewidth]{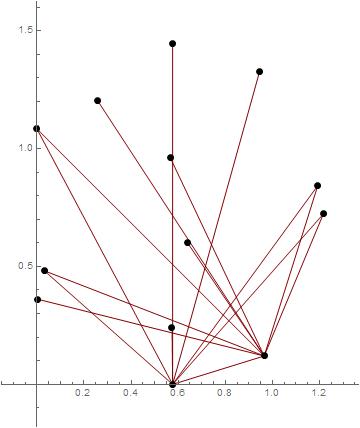}
    \caption{Graph on two-dimensional coordinate system}
  \end{subfigure}
  \caption{The Spiral Embedding of $3+2i$ on Two or Three-Dimensional Coordinate System.}
  \label{fig:graph}
\end{figure}

\begin{figure}[h!]
  \centering
  \begin{subfigure}[b]{0.2\linewidth}
    \includegraphics[width=\linewidth]{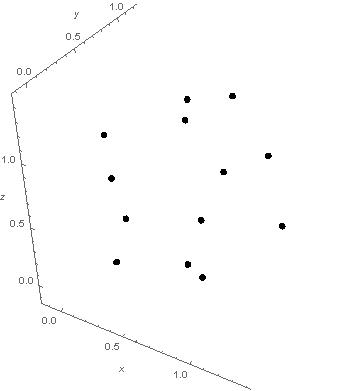}
    \caption{Points on three dimensional coordinate system}
  \end{subfigure}
  \begin{subfigure}[b]{0.2\linewidth}
    \includegraphics[width=\linewidth]{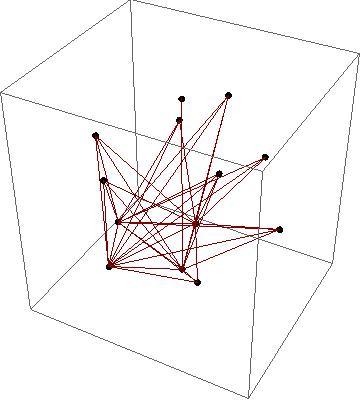}
    \caption{Graph in a three-dimensional box}
  \end{subfigure}
  \begin{subfigure}[b]{0.2\linewidth}
    \includegraphics[width=\linewidth]{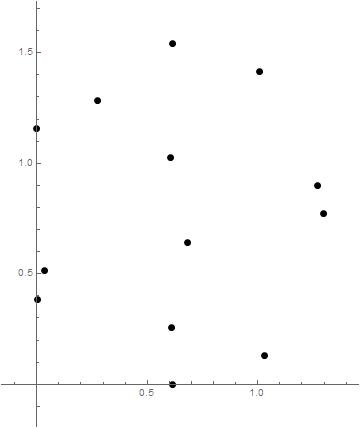}
    \caption{Points on two-dimensional coordinate system}
  \end{subfigure}
  \begin{subfigure}[b]{0.2\linewidth}
    \includegraphics[width=\linewidth]{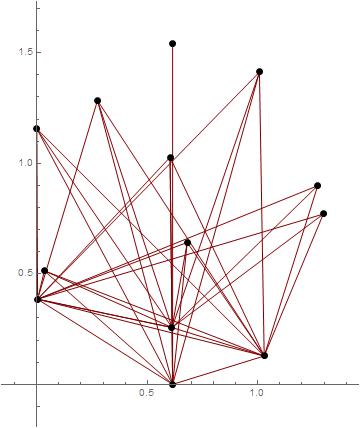}
    \caption{Graph on two-dimensional coordinate system}
  \end{subfigure}
  \caption{The Spiral Embedding of $\frac{5}{2}+\frac{3}{2}i+\frac{3}{2}j+\frac{3}{2}k$ on Two or Three-Dimensional Coordinate System.}
  \label{fig:graph}
\end{figure}

\newpage

\begin{figure}[h!]
  \centering
  \begin{subfigure}[b]{0.2\linewidth}
    \includegraphics[width=\linewidth]{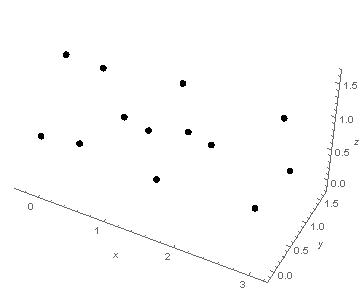}
    \caption{Points on three dimensional coordinate system}
  \end{subfigure}
  \begin{subfigure}[b]{0.2\linewidth}
    \includegraphics[width=\linewidth]{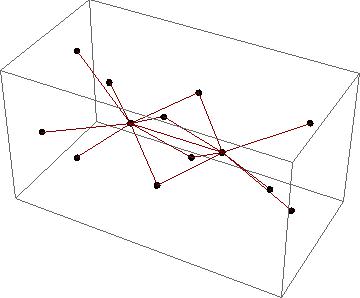}
    \caption{Graph in a three-dimensional box}
  \end{subfigure}
  \begin{subfigure}[b]{0.2\linewidth}
    \includegraphics[width=\linewidth]{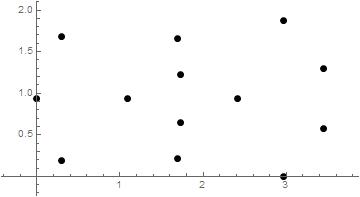}
    \caption{Points on two-dimensional coordinate system}
  \end{subfigure}
  \begin{subfigure}[b]{0.2\linewidth}
    \includegraphics[width=\linewidth]{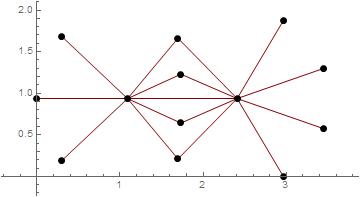}
    \caption{Graph on two-dimensional coordinate system}
  \end{subfigure}
  \caption{Points and Graph of $3+2i$ on Two or Three-Dimensional Coordinate System.}
  \label{fig:graph}
\end{figure}

\begin{figure}[h!]
  \centering
  \begin{subfigure}[b]{0.2\linewidth}
    \includegraphics[width=\linewidth]{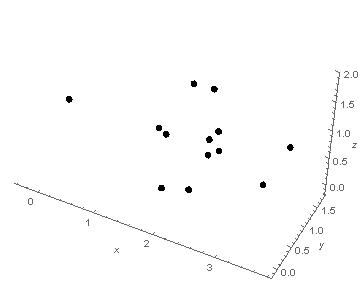}
    \caption{Points on three dimensional coordinate system}
  \end{subfigure}
  \begin{subfigure}[b]{0.2\linewidth}
    \includegraphics[width=\linewidth]{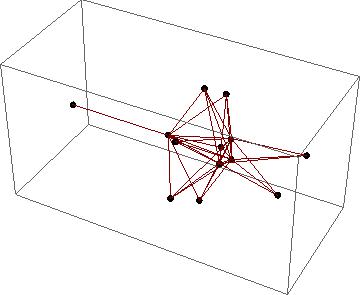}
    \caption{Graph in a three-dimensional box}
  \end{subfigure}
  \begin{subfigure}[b]{0.2\linewidth}
    \includegraphics[width=\linewidth]{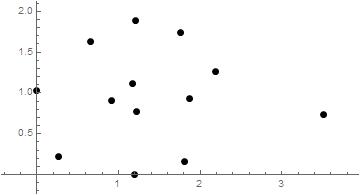}
    \caption{Points on two-dimensional coordinate system}
  \end{subfigure}
  \begin{subfigure}[b]{0.2\linewidth}
    \includegraphics[width=\linewidth]{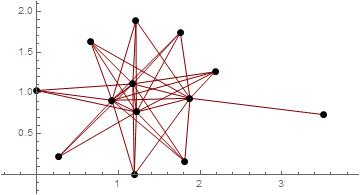}
    \caption{Graph on two-dimensional coordinate system}
  \end{subfigure}
  \caption{Points and Graph of $\frac{5}{2}+\frac{3}{2}i+\frac{3}{2}j+\frac{3}{2}k$ on Two or Three-Dimensional Coordinate System.}
  \label{fig:graph}
\end{figure}

\section{Conclusion} \label{sec8}

In this study, we presented an algebraic construction technique, a modulo function formed depending on two modulo operations, used to construct code constructions over Hurwitz integers. Also, we presented some results for two mathematical notations and, for the algebraic construction technique in this study. In addition, we obtained some new block codes over Hurwitz integers with respect to the modulo function defined in the \hyperref[def10]{definition \ref{def10}} (See \hyperref[sec5]{Section \ref{sec5}}). Moreover, we obtained $(3,1)-$code, and $(4,1)-$code (See \hyperref[sec6]{Section \ref{sec6}}). Lastly, we presented graphs of the residue class obtained with respect to the modulo function defined in the \hyperref[def10]{definition \ref{def10}} in the ring of Hurwitz integers, particularly prime Hurwitz integers.

\bibliographystyle{elsarticle-num}

\begin{thebibliography}{9}

\bibitem{bib1}
Huber, K.: Codes over Gaussian integers. IEEE Transactions on Information Theory, volume 40, issue 1, pp. 207-216, (1994).\href{https://doi.org/10.1109/18.272484}{ doi: 10.1109/18.272484}.

\bibitem{bib2}
Huber, K.: Codes over Eisenstein-Jacobi integers. Finite fields: theory, applications, and algorithms, volume 168, pp. 165–179, (1994).\href{http://dx.doi.org/10.1090/conm/168}{ doi: 10.1090/conm/168}.

\bibitem{bib3}
Freudenberger, J., Ghaboussi, F., Shavgulidze, S.: New coding techniques for codes over Gaussian integers. IEEE Transactions on Communications, volume 61, issue 8, pp. 3114-3124, (2013). \href{https://doi.org/10.1109/TCOMM.2013.061913.120742}{doi: 10.1109/TCOMM.2013.061913.120742}.

\bibitem{bib4}
Ozen, M., Guzeltepe, M.: Codes over quaternion integers. European Journal of Pure and Applied Mathematics, volume 3, issue 4, pp. 670-677, (2010).

\bibitem{bib5}
Ozen, M., Guzeltepe, M.: Cyclic codes over some finite quaternion integer rings. Journal of the Franklin Institute, volume 348, issue 7, pp. 1312-1317, (2011). \href{https://doi.org/10.1016/j.jfranklin.2010.02.008}{doi: 10.1016/j.jfranklin.2010.02.008}.

\bibitem{bib6}
Shah, T., Rasool, S.S.: On codes over quaternion integers. Applicable Algebra in Engineering, Communication and Computing, volume 24, issue 6, (2013). \href{https://doi.org/10.1007/s00200-013-0203-2}{doi: 10.1007/s00200-013-0203-2}.

\bibitem{bib7}
Freudenberger, J., Shavgulidze, S.: New four-dimensional signal constellations from Lipschitz integers for transmission over the Gaussian channel. IEEE Transactions on Communications, volume 63, issue 7, pp. 2420-2427, (2015). \href{https://doi.org/10.1109/TCOMM.2015.2441691}{doi: 10.1109/TCOMM.2015.2441691}.

\bibitem{bib8}
Guzeltepe, M.: Codes over Hurwitz integers. Discrete Mathematics, volume 313, issue 5, pp. 704-714, (2013). \href{https://doi.org/10.1016/j.disc.2012.10.020}{doi: 10.1016/j.disc.2012.10.020}.

\bibitem{bib9}
Rohweder, D., Stern, S., Fischer, R.F.H., Shavgulidze, S., Freudenberger, J.: Four-Dimensional Hurwitz Signal Constellations, Set Partitioning, Detection, and Multilevel Coding, in IEEE Transactions on Communications, volume 69, issue 8, pp. 5079-5090, (2021). \href{https://doi.org/10.1109/TCOMM.2021.3083323}{doi: 10.1109/TCOMM.2021.3083323}.

\bibitem{bib10}
Guzeltepe, M.: On some perfect codes over Hurwitz integers. Mathematical Advances in Pure and Applied Sciences, volume 1, issue 1, pp. 39-45, (2018).

\bibitem{bib11}
Guzeltepe, M., Altınel, A.: Perfect 1-error-correcting Hurwitz weight codes, Mathematical Communications, volume 22, issue 2, pp. 265-272, (2017).

\bibitem{bib12}
Guzeltepe, M., Heden, O.: Perfect Mannheim, Lipschitz and Hurwitz weight codes. Mathematical Communications, volume 19, issue 2, pp. 253-276, (2014).

\bibitem{bib13}
Davidoff, G., Sarnak, P.: Valette, A.: Elementary Number Theory, Group Theory, and Ramanujan Graphs. Cambridge University Press, 2003.

\end{thebibliography}

\end{document}